\documentclass[journal,twocolumn]{IEEEtran}
\IEEEoverridecommandlockouts

\usepackage{epsfig, amsthm,amsmath,amssymb,epsf,cite,scalefnt,epstopdf,setspace}
\usepackage{bbm}
\usepackage[ruled]{algorithm2e}
\usepackage{graphicx}
\usepackage{subfig}
\usepackage{color}
\usepackage{url}
\usepackage{mathtools}
\usepackage{cite}
\usepackage{comment}

\makeatletter
\renewcommand{\maketag@@@}[1]{\hbox{\m@th\normalsize\normalfont#1}}%
\makeatother

\definecolor{mygray}{gray}{0.6}

\def\dW{\boldsymbol{\mathcal{W}}}
\def\dX{\boldsymbol{\mathcal{X}}}
\def\dY{\boldsymbol{\mathcal{Y}}}

  \def\cC{{\mathcal{C}}} 
 \def\cF{{\mathcal{F}}}  
   
\def\cM{{\mathcal{M}}}  \def\cO{{\mathcal{O}}}

 \def\cZ{{\mathcal{Z}}}

\def\ba{{\mathbf{a}}} \def\bb{{\mathbf{b}}} \def\bc{{\mathbf{c}}}  
 \def\bg{{\mathbf{g}}}   
    
\def\bp{{\mathbf{p}}}   \def\bs{{\mathbf{s}}} 
 \def\bv{{\mathbf{v}}} \def\bw{{\mathbf{w}}} \def\bx{{\mathbf{x}}} \def\by{{\mathbf{y}}}
\def\bz{{\mathbf{z}}} 

\def\bA{{\mathbf{A}}}  \def\bC{{\mathbf{C}}}  
 \def\bG{{\mathbf{G}}} \def\bH{{\mathbf{H}}} \def\bI{{\mathbf{I}}}

  \def\bW{{\mathbf{W}}}  \def\bY{{\mathbf{Y}}}

\def\argmin{\mathop{\mathrm{argmin}}}
\def\argmax{\mathop{\mathrm{argmax}}}

\def\tr{\mathop{\mathrm{tr}}}
\def\real{\mathop{\mathrm{Re\!}}}

\def\diag{\mathop{\mathrm{diag}}}

\renewcommand{\baselinestretch}{1.0}

     \def\d4{\!\!\!\!}

 \def\bpsi{\boldsymbol{\mathop{\mathrm{\psi}}}}

  \def\R{{\mathbb{R}}} \def\C{{\mathbb{C}}}

\def\lp{\left(}     \def\rp{\right)}

  \def\-{\! - \!}  \def\+{\! + \!}  \def\={\! = \!}  \def\>{\! > \!}

\def \log{\mathrm{log}}
\def\exp{\mathrm{exp}}

\def \rh{\mathrm{H}}
\def \rt{\mathrm{T}}

\newtheorem{proposition}{Proposition}

\newtheorem{lemma}{Lemma}

\newtheorem{remark}{Remark}

\newcommand{\bef}{\begin{figure}}
\newcommand{\eef}{\end{figure}}
\newcommand{\beq}{\begin{eqnarray}}
\newcommand{\eeq}{\end{eqnarray}}

\def\rh{\mathrm{H}}
\def\rt{\mathrm{T}}

\def\bOmega{\boldsymbol{\Omega}}

\def\bpsi{\boldsymbol{\psi}}

\def\treq{\triangleq} 
\def\tr{\mathop{\mathrm{tr}}}
\def\real{\mathop{\mathrm{Re}}}
\def\imag{\mathop{\mathrm{Im}}}

\title{Localization and Discrete Beamforming with a Large Reconfigurable Intelligent Surface}

\author{Baojia Luo,\thanks{B. Luo, Y. Deng, and Z. Huang are with Department of Mathematical Sciences, Tsinghua University, Beijing, China (e-mails:	\{ luobj19,dengyl21\}@mails.tsinghua.edu.cn, zhongyih@tsinghua.edu.cn).} Yili Deng, Miaomiao Dong, Zhongyi Huang,\thanks{M. Dong, X. Chen, W. Han, and B. Bai are with Theory Lab, Central Research Institute, 2012 Labs, Huawei Technologies Co. Ltd., Hong Kong, China (e-mails: \{dong.828599, chenxiang73, harvey.hanwei, baibo8\}@huawei.com).} Xiang Chen, Wei Han, and Bo Bai
\thanks{This work was partially supported by the National Natural Science Foundation of China (No. 12025104).
Parts of this paper have been presented at the IEEE International Conference on Communications 2022 \cite{luoReconfigurableIntelligentSurface2022}.}
\vspace{-0.0cm}}

\begin{document}
\maketitle

\begin{abstract}
In millimeter-wave (mmWave)  cellular systems, reconfigurable intelligent surfaces (RISs)  are foreseeably deployed with a large number of reflecting elements to achieve high beamforming gains.
The large-sized RIS will make radio links fall in the near-field localization regime with spatial non-stationarity issues.
Moreover, the discrete phase restriction on the RIS reflection coefficient incurs exponential complexity for discrete beamforming.
It remains an open problem to find the optimal RIS reflection coefficient design in polynomial time.
To address these issues, we propose a scalable partitioned-far-field protocol that considers both the near-filed non-stationarity and discrete beamforming.
The protocol approximates near-field signal propagation using a partitioned-far-field representation to inherit the sparsity from the sophisticated far-field and facilitate the near-field localization scheme. 
To improve the theoretical localization performance, we propose a fast passive beamforming (FPB) algorithm that optimally solves the discrete RIS beamforming problem, reducing the search complexity from exponential order to linear order.
Furthermore, by exploiting the partitioned structure of RIS, we introduce a two-stage coarse-to-fine localization algorithm that leverages both the time delay and angle information. 
Numerical results demonstrate that centimeter-level localization precision is achieved under medium and high signal-to-noise ratios (SNR), revealing that RISs can provide support for low-cost and high-precision localization in future cellular systems.

\end{abstract}

\begin{IEEEkeywords}
Reconfigurable intelligent surface, millimeter wave, discrete phase shift, optimal passive beamforming, near-field localization. \vspace{-0mm}
\end{IEEEkeywords}

 \vspace{-0mm}


\section{Introduction}~\label{sec:intro}
The fifth-generation (5G) cellular system
has leveraged the broadband millimeter wave (mmWave)  spectrum to provide over 10~Gbit/s peak data rate \cite{gao2021mimo}. 
In mmWave communications, large-sized antenna arrays are installed at the transceiver to combat severe high-frequency pathloss using directional beams \cite{gao2023integrated}. 
These pencil-like beams are signatured by different angles in space and have high angular resolutions, which reveal a potential of high-precision localization using the angle-of-arrival (AoA) and/or angle-of-departure (AoD) 
\cite{7999215,elzanaty6GHolographicLocalization2021a,ditarantoLocationAwareCommunications5G2014}.
Nevertheless, mmWave links are vulnerable to physical blockage due to high-frequency and directional transmission. 
To provide reliable and seamless coverage for location-aware services in mmWave networks, 
a large number of mmWave base stations (BSs) should be installed. 
This dense deployment, coupled with the large-sized antenna arrays at the BSs, results in high expenditures on the mmWave network planning and operation.

A promising solution to address the above issues is the reconfigurable intelligent surface (RIS), a two-dimensional plane comprising a large number of sub-wavelength reflecting units \cite{direnzoSmartRadioEnvironments2020,bjornsonReconfigurableIntelligentSurfaces2022, wuIntelligentReflectingSurfaceAided2021}.
These reflection units use positive-intrinsic-negative (PIN) diodes or varactors to adjust their reflection coefficients. 
By reasonably configuring the reflection coefficients, RISs can reflect signals with directional high-gain beams \cite{elzanatyReconfigurableIntelligentSurfaces2021}. The beamforming gain can effectively improve the radio link quality and increase as the size of RIS grows.  
Due to the removal of radio frequency components and passive reflection, large-sized RISs can be economically fabricated and then deployed on walls or roofs to replace the mmWave BSs  \cite{basarWirelessCommunicationsReconfigurable2019a}.
Therefore, a large-sized RIS assisted mmWave system is a cost-effective solution for achieving high-precision localization in 5G.

To fully unleash the potential of RIS,
one should optimally configure the RIS reflection coefficients (referred to do \textit{passive beamforming}) based on the channel state information.
Because of the large-sized RIS and short mmWave wavelength, the electromagnetic propagation environment of RIS is changed from the traditional far field to the near field \cite{JiangGao2022RIS,cuiChannelEstimationExtremely2022}.
Different from the channel model in the far field, there exists the 
spatial non-stationarity in the near field so that different RIS elements observe different channel multipath characteristics. 
This non-stationarity in the near-field propagation leads to more complicated channel modeling and thus 
requires a new philosophy in the design of passive beamforming and localization algorithms.
Existing works  \cite{alexandropoulosLocalizationMultipleReconfigurable2022a,wangJointBeamTraining2021,fascistaRISaidedJointLocalization2022,tengBayesianUserLocalization2022a,linChannelEstimationUser2021} 
that design the RIS reflection coefficients for localization in the far field cannot be directly applied in the near field.

Moreover, the hardware limitation of RIS  causes a scalability issue for the passive beamforming.
In practical implementation,  the PIN diode or varactor of each RIS unit is operated on a finite number of voltages, leading to discrete reflection coefficient settings. 
The discrete constraints make the RIS passive beamforming a combinatorial optimization problem whose solving complexity exponentially grows with the size of RIS. 
This computational complexity is unaffordable for a large-sized RIS, and one has to limit the sizes of RIS to find a tradeoff between the passive beamforming gain and computational complexity \cite{linReconfigurableIntelligentSurfaces2021,9374451,luoSpatialModulationRISAssisted2021}.  
In a large-sized RIS assisted localization protocol, it is imperative to design a high-gain and low-complexity passive beamforming algorithm.

 \vspace{-0mm}

\subsection{Related Work}~\label{subsec:Related Work}
 \vspace{-0mm}
 
In RIS-assisted mmWave systems, extensive research efforts \cite{rinchiCompressiveNearFieldLocalization2022,keykhosraviRISEnabledSISOLocalization2022, hanLocalizationChannelReconstruction2022,zhangMetaLocalizationReconfigurableIntelligent2021a,dardariNLOSNearFieldLocalization2021,huangNearFieldRSSBasedLocalization2022} have been dedicated to addressing the channel non-stationarity for the near-filed localization. 
Compressive sensing (CS)-based approaches were investigated to exploit the sparsity of near-field channels in \cite{rinchiCompressiveNearFieldLocalization2022,keykhosraviRISEnabledSISOLocalization2022, hanLocalizationChannelReconstruction2022}.  A received signal strength (RSS)-based approach in \cite{zhangMetaLocalizationReconfigurableIntelligent2021a} configures the  RIS reflection coefficients to generate a preferable RSS distribution and then utilizes the measured RSS at each user equipment (UE) for localization.  
Authors in \cite{rinchiCompressiveNearFieldLocalization2022,keykhosraviRISEnabledSISOLocalization2022, hanLocalizationChannelReconstruction2022,zhangMetaLocalizationReconfigurableIntelligent2021a} design new localization algorithms tailored for the near-field channel model.  
To use existing localization algorithms in the far field, the RIS partitioning technique has recently been proposed.
This technique partitions a large RIS into several small-sized sub-surfaces, each falling in the far field. 
Using the RIS partitioning technique, the time difference-of-arrival (TDoA) based localization was investigated in a   single-input single-output (SISO) Orthogonal Frequency Division Multiplexing (OFDM)
system \cite{dardariNLOSNearFieldLocalization2021}, and the AoD-based 
localization was studied in a multiple-input multiple-output (MIMO) narrow band system \cite{huangNearFieldRSSBasedLocalization2022}.
The \cite{dardariNLOSNearFieldLocalization2021} and \cite{huangNearFieldRSSBasedLocalization2022} exploit either the TDoA or AoD to estimate the UE position and achieve decimeter-level accuracy.
In the near-filed MIMO-OFDM system, there is a potential to jointly exploit the time delay and angle information to achieve centimeter-level accuracy with the assistance of RISs. 

In RIS-assisted mmWave localization systems, 
configuring the RIS reflection coefficients contributes to localization accuracy.
Theoretically, the Cram\'er-Rao Lower Bound (CRLB) for localization error can be derived as a function of the RIS reflection coefficients \cite{fascistaRISaidedJointLocalization2022}. 
However, the CRLB expression is often complicated, and it is intractable to design the RIS reflection coefficients by minimizing the CRLB. 
To circumvent this difficulty, the RIS configuration is typically investigated by solving a discrete-constrained passive beamforming problem (referred to \textit{discrete beamforming problem}) that maximizes the passive beamforming gain 
\cite{heLargeIntelligentSurface2019,bjornsonReconfigurableIntelligentSurfaces2022}. 
Efficient algorithms to solve the discrete beamforming problem have attracted substantial research attention \cite{anLowComplexityChannelEstimation2022,diHybridBeamformingReconfigurable2020,wuBeamformingOptimizationWireless2020,liJointDesignHybrid2022,zhangConfiguringIntelligentReflecting2022,sanchezOptimalLowComplexityBeamforming2021,ren2022linear,xiong2023optimal,pekcan2023achieving}.
As a combinatorial optimization problem, the discrete beamforming problem can be optimally solved by an exhaustive search  \cite{anLowComplexityChannelEstimation2022} or the branch-and-bound search \cite{diHybridBeamformingReconfigurable2020} on different RIS configurations, neither of which is computationally scalable for a large-sized RIS. 
To achieve the scalability, heuristic approaches \cite{wuBeamformingOptimizationWireless2020,liJointDesignHybrid2022} have been employed to obtain sub-optimal solutions.
Recently, several polynomial-complexity optimal methods have been proposed in \cite{zhangConfiguringIntelligentReflecting2022, sanchezOptimalLowComplexityBeamforming2021,ren2022linear,xiong2023optimal,pekcan2023achieving} to solve the discrete beamforming problem under simplified settings.
Specifically, the binary-phase beamforming strategy in \cite{zhangConfiguringIntelligentReflecting2022} focuses on the binary-phase setting of RIS with reflection coefficients selected from set $\{e^{j0},e^{j\pi}\}$, resulting in linear complexity, while the SISO beamforming strategy in \cite{sanchezOptimalLowComplexityBeamforming2021,ren2022linear,xiong2023optimal,pekcan2023achieving} consider only a single-antenna receiver.
However, neither of these methods can simultaneously handle the increased number of the RIS unit's phase shifts and the coupled inference of multi-antenna at the receiver.
Up to now, a polynomial-time algorithm for solving the optimal discrete beamforming problem in MIMO systems with general discrete RIS reflection coefficient settings is still lacking.

 \vspace{-0mm}

\subsection{Overview of Methodology and Contributions}~\label{subsec:Contributions}
 \vspace{-0mm}
 
In this paper, we investigate the UE localization in a RIS-assisted MIMO-OFDM system operating at mmWave bands. 
Sizes of the RIS are assumed to be large so that the RIS links are in the near-field. 
By approximating the near field to the partitioned far field, we devise a low-complexity near-field localization protocol including an efficient signal components separation method, an optimal and linear-time
discrete beamforming method, and a high-precision two-stage positioning method.
The main contributions of this paper are summarized as follows:
\begin{itemize}
\item  
We develop a balanced signaling method to separate the UE received signals into the non-line-of-sight (NLoS) component generated by scatters and the line-of-sight (LoS) component reflected by the RIS. 
The localization problem and the CRLB are formulated based on the separated LoS component.
Leveraging the balanced configuration of RIS, we simplify both the localization problem and the CRLB and derive the theoretical relationship between the CRLB and the RIS beamforming gain.

\item We propose a fast
passive beamforming (FPB) algorithm to optimally solve the discrete passive beamforming problem and reduce the search complexity from exponential order to linear order.
The FPB algorithm addresses the scalability issue in large-scale RIS configuration 
and improves the localization performance.
We provide theoretical analysis to demonstrate the optimality of the FPB algorithm. 
Our approach is applicable to the MIMO scenario and general discrete reflection coefficient settings of RIS.

\item We jointly exploit the time-of-arrival (ToA) and AoA information of the partitioned-far-field approximated channel and propose a two-stage coarse-to-fine localization algorithm that efficiently solves the non-convex localization problem.
Finally, we evaluate the efficacy of the proposed localization protocol and numerically show that the proposed scheme can achieve centimeter-level localization accuracy in medium and high signal-to-noise ratio (SNR) regimes.

\end{itemize}

The remainder of this paper is organized as follows.
In Section \ref{sec:System_Model}, the system model and partitioned-far-field representation are described.
In Section \ref{sec: protocol}, the overall framework of the localization protocol is presented.
The balanced signaling approach and the localization problems are formulated in Section \ref{sec: problem}.
The position error bounds are derived, and the FPB algorithm for RIS configuration is presented in Section \ref{sec: FIM}.
In Section \ref{sec:LocAlgo}, a two-stage coarse-to-fine localization algorithm is proposed. 
Simulation results and conclusions are given in Section \ref{simulation_sec} and \ref{Conclusions}, respectively.

\emph{Notations:}  $\ba$, $\bA$, $\mathcal{A}$ and $\boldsymbol{\mathcal{A}}$ denote vectors, matrices, sets, and tensors, respectively.
$\C$ and $\R$ represent the complex and real number sets.
$(\cdot)^*$, $(\cdot)^\rt$, and $(\cdot)^\rh$ denote the conjugate, transpose, and Hermitian transpose, respectively.
$\Vert \ba \Vert_2$ and $\left\|\dX \right\|_F$ denote the Euclidean norm and Frobenius norm.
$\tr(\bA)$, $\bA^{-1}$, $\mathrm{rank}(\bA)$ and $\mathrm{diag(\bA)}$ denote the trace of $\bA$, the inverse of $\bA$, the rank of $\bA$, and the diagonal elements of $\bA$, respectively.
The $(i,j,k)$-th entry of three-oreder tensor $\dX$ is denoted by $\dX_{i,j,k}$ while the vector $\ba\=\dX_{i,j,:}$ is obtained by stacking $\dX_{i,j,k}$ over index $k$ for fixed $i$-th row and $j$-th column.
The same definitions are applicable to vectors $\dX_{:,j,k}$ and $\dX_{i,:,k}$.
For a complex number $z \in \C$, $\mathrm{arg}(z)$ refers to the principal argument of $z$, $\real(z)$ the real part, $\imag(z)$ the imaginary part, and $\left| z \right|$ the modulus. 
${\boldsymbol 1}_N $, ${\boldsymbol 0}_{M\times N} $ and $\bI_N$ denote all-one vectors with size $N\times 1$, zero matrices with size $M\times N$ and identity matrices with size $N\times N$, respectively. 
$\odot$ and $\otimes$ denote the Hadamard and Kronecker products. 
$\mathcal{C}\mathcal{N}(\boldsymbol \mu,\bC)$ refers to the distribution of a circularly symmetric complex Gaussian random vector with mean vector $\boldsymbol \mu$ and covariance matrix ${\bC}$. 
Finally, $\mathbb{E}(\cdot)$ denotes the mathematical expectation while $\mathrm{Cov}(\cdot)$ denotes the auto-covariance operator.

\section{System Model}~\label{sec:System_Model}
In this section, we first describe the RIS-assisted localization mmWave system and then introduce the 
channel and signal models used in subsequent derivations. 
\vspace{-0mm}
\subsection{System Setup}

As illustrated in Fig.~\ref{fig:model_illustration}, we consider a RIS-assisted localization system consisting of one BS, UE, and RIS, each in a two-dimensional scenario. 
The localization system operates on a mmWave band with central frequency $f_c$ and bandwidth $B$. 
The bandwidth $B$ is equally divided into $N$ subcarriers with frequency spacing  $\Delta f$. 
The frequency at the $n$-th subcarrier is represented by $f_n = f_c + (n-(N+1)/2)\Delta f$, where $n \in \{1,2,\ldots,N\}$. Antennas at the BS and UE are uniform linear arrays (ULAs) with element numbers $N_\text{T}$ and $N_\text{R}$, respectively. 
For $i \in \{1, 2, \dots, N_\text{R}\}$ and $j \in \{1, 2, \dots,N_\text{T}\}$, the $i$-th receive antenna is located in a unknown position $\bp^{(i)}$, and the $j$-th transmit antenna of BS is located in a known position $\bp_\text{T}^{(j)}$. 
For ease of exposition, we construct a Cartesian coordinate system where the reference point is at the first antenna of BS, i.e., $\bp_\text{T}^{(1)}=[0,0]^\rt$. 
Since the position and orientation information of BS and RIS are all known, without loss of generality, we assume that the orientations of  the BS, UE, and RIS are in parallel.

The RIS contains $M$ elements with half-wavelength spacing $\lambda/2$, and for ease of exposition, the $M$ elements are assumed to reside in one line forming an $M$-element ULA, as illustrated in Fig.~\ref{fig:model_illustration}. 
The $m$-th element is located in a known position $\bp_\text{M}^{(m)}, m \in \{1, 2, \dots,M \}$.  
The RIS unit can adjust its reflection coefficient to customize the wireless propagation environment. 
When a certain voltage is applied to the $m$-th RIS element  at time  $t \in \{1, 2, \dots, T \}$, its reflection coefficient is set as $\beta_{m,t} \exp(j \omega_{m,t})$, where 
$\beta_{m,t}$ is the reflection amplitude, and $\omega_{m,t}$ is the adjustable phase shift. 
It is commonly assumed that $\beta_{m,t} \= 1$ and $ \exp(j \omega_{m,t})$ takes values from a finite discrete set $\mathcal{F} \= \{ \exp(j2\pi \cdot s / 2^b ), s \= 0,1,\ldots,2^b - 1 \}$, where 
$b\!\ge\!1$ refers to the number of bits for discrete phase adjustment.

 \begin{figure}
\centering
\includegraphics[width=0.4\textwidth]{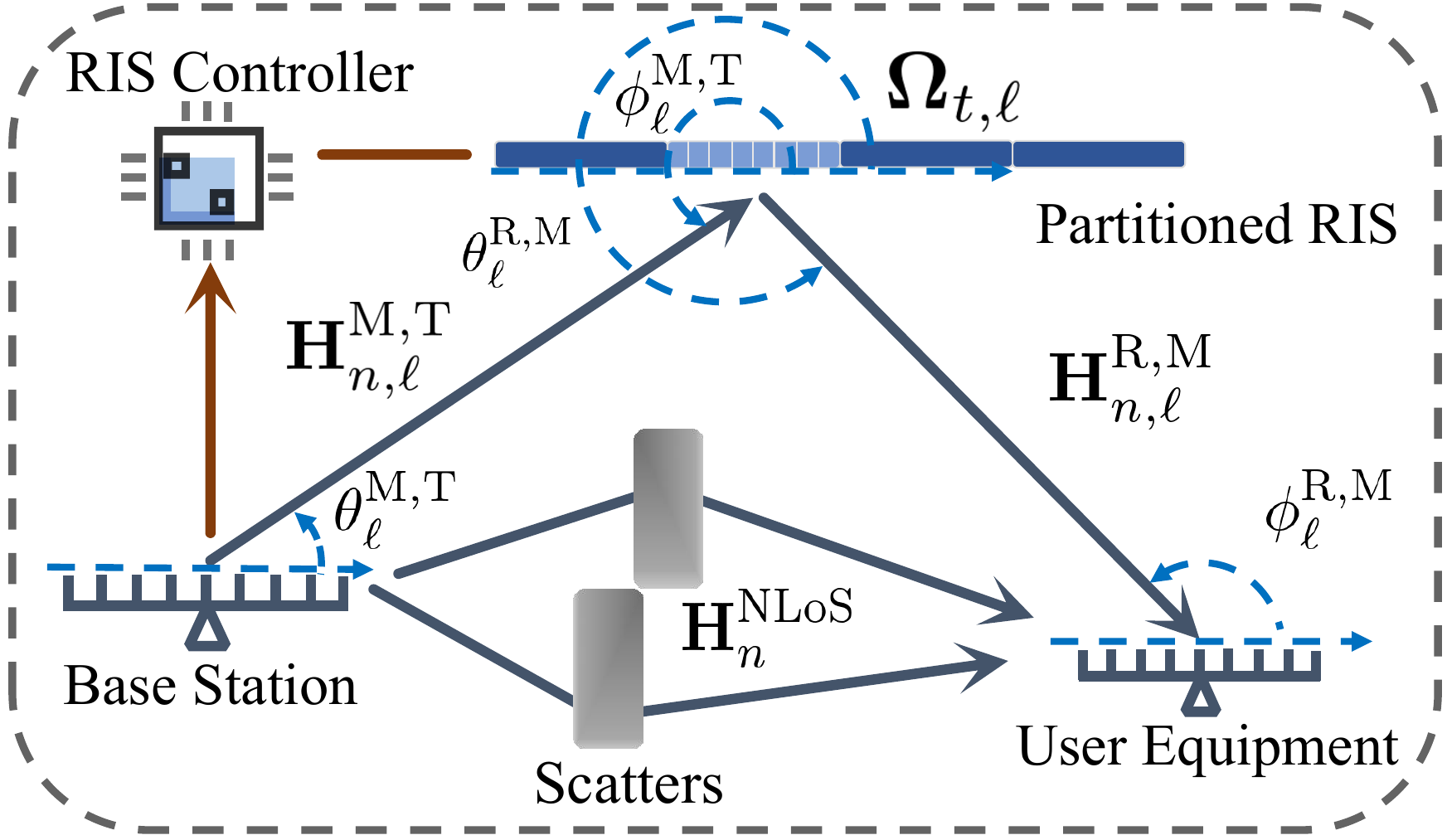}
\caption{RIS-assisted localization system.}
\label{fig:model_illustration}
\end{figure}

\subsection{Near-field Channel Model}

We consider a downlink localization system, where the UE estimates its position based on received signals from the BS. Since mmWave signals suffer from severe attenuation, in the downlink transmission, we only consider single-bounced mmWave links. This means we consider the scenario where the BS-RIS channel $\bH_n^\text{M,T} \!\in\! \C^{M \times N_\text{T}}$ and the RIS-UE channel $\bH_n^\text{R,M} \!\in\! \C^{N_\text{R} \times M}$ contain only the LoS component, and the direct channel $\bH_n^\text{NLoS}\!\in\! \C^{N_\text{R} \times N_\text{T}}$ between the BS and UE contains NLoS components.
The overall pilots' transmission time $T$ is assumed to be shorter than the channel coherence time.
Therefore, on the $n$-th subcarrier, channels $\bH_n^\text{M,T}$, $\bH_n^\text{R,M}$ and $\bH_n^\text{NLoS}$ are invariant to the time slot $t$.
The channel for subcarrier $n$ at time $t$ is then indicated as
\begin{equation}
\small
\bH_n(t) = \bH_n^\text{R,M} \boldsymbol \Omega_t \bH_n^\text{M,T} + 
\bH_n^\text{NLoS} \in \C^{N_\text{R} \times N_\text{T}}, 
\label{eqn: channel_all}
\end{equation}
where $\boldsymbol \Omega_t = \diag \left\lbrace \exp(j\omega_{1,t}),\exp(j\omega_{2,t}),\ldots,\exp(j\omega_{M,t}) \right\rbrace\in \C^{M \times M}$ represents the time-varying reflection coefficient matrix of RIS. 

As aforementioned in Section~\ref{sec:intro}, a large-sized RIS (i.e., large $M$) will yield near-field channels $\bH_n^\text{M,T}$ and $\bH_n^\text{R,M}$, leading to spatial non-stationarity.
Specifically, for $i =1, 2, \dots,N_\text{R}$ and $j = 1, 2, \dots,N_\text{T}$, the RIS reflected channel in \eqref{eqn: channel_all} between the $i$-th receive antenna of UE and the $j$-th transmit antenna of BS is
\begin{equation}\small
\begin{aligned}
&\left[ \bH_n^\text{R,M} \bOmega_t \bH_n^\text{M,T}\right]_{i,j}  \\  
& = \!\sum_{m=1}^M \!\sqrt{ \rho_{m,j}\rho_{m,i}}  \exp\big\lbrace \!-\!j2\pi f_n\left( \tau_{m,j} +\tau_{m,i} \right) \!+\! j \omega_{m,t}  \big\rbrace ,
\end{aligned}
\label{eqn: los_scalar}
\end{equation}
where $\tau_{m,j} = \Vert \bp_\text{M}^{(m)} - \bp_\text{T}^{(j)}\Vert_2/c$ is the ToA between the $j$-th transmit antenna of BS and the $m$-th element of RIS, and $\tau_{m,i} = \Vert \bp_\text{M}^{(m)}  - \bp^{(i)}\Vert_2/c$ is the ToA between the $i$-th receive antenna of UE and the $m$-th element of RIS. 
$c$ is the speed of light. 
$ \rho_{m,j}= \Vert \bp_\text{M}^{(m)} - \bp_\text{T}^{(j)}\Vert_2^{-\mu}$ and $\rho_{m,i} = \Vert\bp_\text{M}^{(m)} -\bp^{(i)}\Vert_2^{-\mu}$ are the  pathloss where $\mu>2$ is the scenario-depended pathloss exponent \cite{heLargeIntelligentSurface2019}. 
As for the remaining NLoS component in \eqref{eqn: channel_all}, we assume the widely-used multipath model
\footnotesize\begin{equation*}
\left[ \bH_n^\text{NLoS} \right]_{i,j} =  \sum_{k=1}^K \sqrt{\rho_{i,j}^{(k)}} \exp\left\lbrace -j2\pi f_n \tau_{i,j}^{(k)}\right\rbrace,
\end{equation*}\normalsize
where $\tau_{i,j}^{(k)}$ and $\rho_{i,j}^{(k)}$ are randomly-generated ToAs and pathloss scattered by $K$ unknown obstacles.

\vspace{-0mm}
\subsection{Partitioned-far-field Representation}
This work uses the ToA and AoA of the RIS reflected channel in \eqref{eqn: los_scalar} for localization. 
Because of the near-field channel structure, it is difficult to directly extract the angles and time delays from the near-field channel. 
To tackle this issue, we next use a partitioned-far-field technique to approximate the near-field channel in \eqref{eqn: los_scalar}.

The basic idea is to equally partition the large-sized RIS into $L\ge 1$ consecutive segments, as illustrated in Fig.~\ref{fig:model_illustration}, and thereby approximate the spherical wave propagation with the partitioned-planar wave. 
Here, the size of each segment, i.e., ${M\lambda}/{(2L)}$, should be small enough to ensure the far-field signal propagation between a partitioned RIS segment and a UE.
Therefore, the criterion is that the minimum signal propagation distance $D$ is greater than the Rayleigh distance $\frac{2}{\lambda} \left( \frac{M}{L}\frac{\lambda}{2}\right) ^2$.
In other words, $L$ is large enough to satisfy 
\begin{equation}
\small
L \ge \sqrt{\frac{\lambda}{2D}}M.
\label{eqn: Fraunhofer}
\end{equation}
Although  $D$ is unknown, we can find an estimate of the $D$ based on the area of interest. 
Once the constraint of $L$ in \eqref{eqn: Fraunhofer} is satisfied, the channel between BS and the $\ell$-th RIS segment $\bH^\text{M,T}_{n,\ell}\in\C^{(M/L)\times N_\text{T}}$ and the channel between UE and the $\ell$-th RIS segment $\bH^\text{R,M}_{n,\ell}\in\C^{N_\text{R}\times (M/L)}$ are both  in the far field. 
This yields a simplified approximation of the cascaded near-field channel $\bH_n^\text{R,M} \boldsymbol \Omega_t \bH_n^\text{M,T}$ in \eqref{eqn: los_scalar}, which is given by 
\begin{equation}\small
\bH_n^\text{R,M} \boldsymbol \Omega_t \bH_n^\text{M,T} =  \sum_{\ell=1}^L  \bH^\text{R,M}_{n,\ell}\boldsymbol \Omega_{t,\ell} \bH^\text{M,T}_{n,\ell},
\label{eqn: segment_ris}
\end{equation}
where $\boldsymbol \Omega_{t,\ell} \!\in\! \C^{(M/L)\times (M/L)}$ is
the reflection coefficient matrix of the $\ell$-th RIS segment. 
Since $\bH^\text{R,M}_{n,\ell}$ and $\bH^\text{M,T}_{n,\ell}$ in \eqref{eqn: segment_ris} are both far-field channels constructed by an LoS path, their angular-domain expressions are given by \cite{shahmansooriPositionOrientationEstimation2018b,dongSimulationStudyMillimeter2015}
\begin{align}
\small \bH^\text{M,T}_{n,\ell}  = \rho^\text{M,T}_{n,\ell} 
\ba_{\frac{M}{L}}\left(  \cos \phi_\ell^\text{M,T}\right)  \ba^\rh_{N_\text{T}} \left( \cos \theta^\text{M,T}_\ell\right) ,
\label{eqn: channel_bs_ris}
\end{align}
\begin{align}
\small \bH^\text{R,M}_{n,\ell} =  \rho^\text{R,M}_{n,\ell} 
\ba_{N_\text{R}}\left( \cos \phi_\ell^\text{R,M}\right)  \ba^\rh_{\frac{M}{L}} \left( \cos \theta^\text{R,M}_\ell\right) ,
\label{eqn: channel_ris_ue}
\end{align}
In \eqref{eqn: channel_bs_ris} and \eqref{eqn: channel_ris_ue}, $\ba_N(f) \=[ 1,\exp (  -j \pi f ),\ldots,\exp (  -j  \pi (N-1) f )]^\rt\!\in\! \C^{N \times 1}$ is the array response.
$\tau_\ell^\text{M,T}$, $\theta^\text{M,T}_\ell$ and $ \phi_\ell^\text{M,T}$ are the ToA, AoD, and AoA from the BS to the $\ell$-th RIS segment, respectively. 
Similar definitions apply to the parameters in the RIS-UE channel including $\tau_\ell^\text{R,M}$, $\theta^\text{R,M}_\ell$, and $ \phi_\ell^\text{R,M}$. 
Besides, $\rho^\text{M,T}_{n,\ell}  = \sqrt{ \rho^\text{M,T}_\ell  } \exp\lbrace -j 2\pi \tau_\ell^\text{M,T} f_n \rbrace$ and $\rho^\text{R,M}_{n,\ell} = \sqrt{ \rho^\text{R,M}_\ell }\exp\lbrace -j 2\pi \tau_\ell^\text{R,M} f_n \rbrace$ are the pathloss under each subcarrier.
\vspace{-0mm}
\subsection{UE Received Information}\label{subsec:Signal_model}

Before the localization process, the position information of the BS and RIS is sent to the UE. 
To accomplish the localization task, the BS sends  pilot signals at each subcarrier $n\=1,2,\ldots, N$ and time slot $t\=1,2,\ldots, T$. 
Under the proposed partitioned-far-field representation in \eqref{eqn: segment_ris}, the received signals over subcarrier $n$ and time $t$ is given by
\begin{equation}
\label{eqn: received_signal}\small
\by_{n,t} = \sum_{\ell=1}^L  \bH^\text{R,M}_{n,\ell}\boldsymbol \Omega_{t,\ell} \bH^\text{M,T}_{n,\ell}\bx_{n,t} + \bH_n^\text{NLoS} \bx_{n,t} + \bw_{n,t},
\end{equation}
where the transmitted pilot $\bx_{n,t} = \sqrt{P_\text{T}} \bv s_{n,t} \in \C^{N_\text{T} \times 1}$ is precoded by the BS beamformer $\mathbf{v}$. 
$P_\text{T}$ is the transmitted power, and $s_{n,t}$ represents the transmit symbol. 
$\bw_{n,t} \sim \mathcal{C}\mathcal{N}(\boldsymbol 0,\sigma^2 \bI_M)$ with noise power $\sigma^2$.
Since the transmitted symbols are known to both the BS and UE, we can assume that the pilot symbols are time-invariant and subcarrier-invariant, i.e., $\bs_{n,t} = \bs$. 
 The UE utilizes the $N \times T$ observations $\{\by_{n,t} \}_{n,t}$ in \eqref{eqn: received_signal} to estimate its own position $\bp$.


\begin{figure}
\centering
\includegraphics[width=0.4\textwidth]{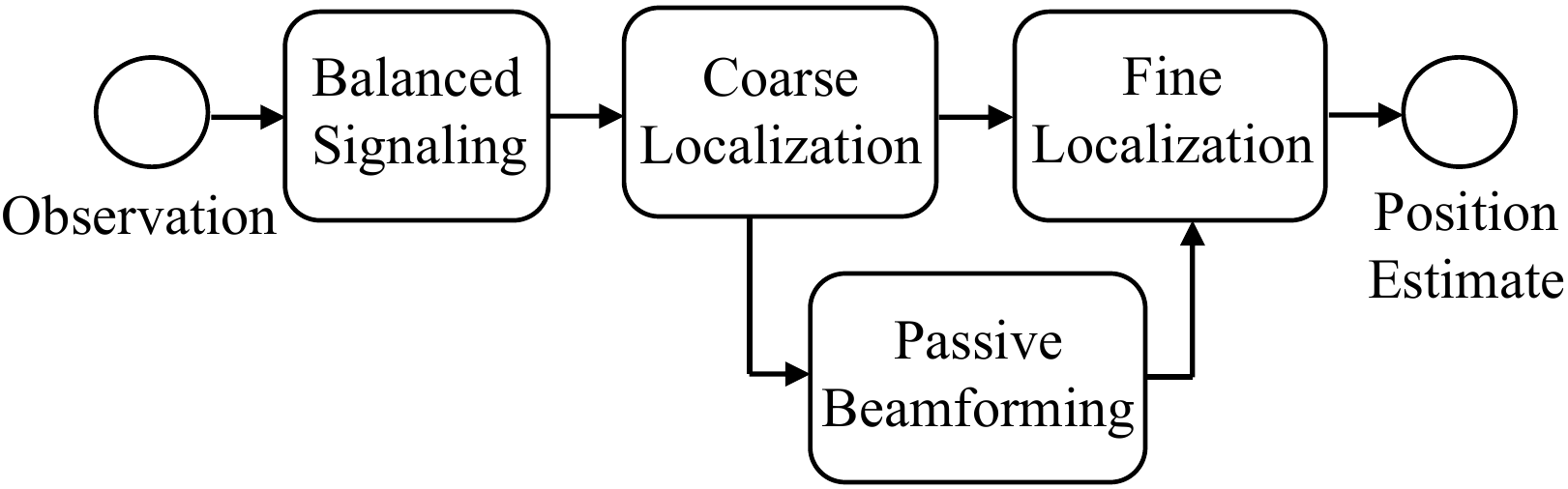}
\caption{Framework of the proposed localization protocol.}
\label{fig: task}
\end{figure}


\section{RIS-aided Localization Protocol}  \label{sec: protocol}
In this section, we schematically describe the proposed RIS-aided localization protocol, which comprises four modules: the \emph{balanced signaling module}, the \emph{coarse localization module}, the \emph{passive beamforming module}, and the \emph{fine localization module}, as sketched in Fig. \ref{fig: task}.

The localization scheme relies on cooperation between the BS and UE.
At the beginning of the localization task, the measurements $\{\by_{n,t}\}_{n,t}$ in \eqref{eqn: received_signal} are first processed by the balanced signaling module to extract the  LoS component at the UE. 
Taking the separated LoS component as the input, in the first $T/2$ time slots, the coarse localization module at the UE localizes itself with low complexity but a relatively large estimation error.
A  coarse localization result is sent to the passive beamforming module at the BS, which designs the optimal RIS reflection coefficient and configures the coefficient of RIS based on the received feedback. 
In the latter $T/2$ time slots, the fine localization module at the UE uses the measurements and the coarse localization result to localize itself with higher accuracy.
The functions of these four modules are sequentially introduced as follows.

\subsubsection{Balanced Signaling Module} 
This module processes the  measurements $\{\by_{n,t}\}_{n,t}$ in \eqref{eqn: received_signal} to separate the NLoS component $\bH_n^{\text{NLoS}} \bx_{n,t} $ and the RIS-reflected LoS component $\sum_{\ell=1}^L  \bH^\text{R,M}_{n,\ell}\boldsymbol \Omega_{t,\ell} \bH^\text{M,T}_{n,\ell}\bx_{n,t} $.
The separation is executed by configuring the time-variant RIS reflection coefficient $\{ \boldsymbol \Omega_{t,\ell} \}_{t,\ell}$ to satisfy the balanced condition, as detailed in Section~\ref{subsec: balanced_method}.
This separation enables us to eliminate interference from the random NLoS component and to use the RIS-reflected LoS component to formulate the localization problem in Section~\ref{subsec: loc_problem}.

\subsubsection{Coarse Localization Module}
Based on the observations at the first $T/2$ realizations, this module extracts the ToA and AoA information and utilizes them to coarsely estimate the UE's position.
The detailed description of this module is provided in Section~\ref{sec:LocCoarse}.
Under the partitioned-far-field representation, the coarse localization module provides multiple estimations of the UE's position, denoted by $\{ \hat{\bp}_\text{ce}^\ell \}_{\ell=1}^L$.
These estimates will be utilized by the passive beamforming module and the fine localization module below.

\subsubsection{Passive Beamforming Module} 
This module designs and configures the RIS reflection coefficients $\{ \boldsymbol \Omega_{t,\ell} \}_{t,\ell}$ to achieve high beamforming gain and improve the localization performance.
As emphasized in Section~\ref{subsec: CRLB} and \ref{subsec: beamform_problem}, the localization error bound, CRLB, can be lowered by solving the discrete-constrained passive beamforming problem (i.e., the discrete beamforming problem).
On this basis, in Section~\ref{subsec: FPB}, we propose an FPB algorithm to optimally solve this combinatorial optimization problem with only linear search complexity.

\subsubsection{Fine Localization Module} 
Based on the observations at all $T$ realizations, this module refines the coarse localization results after the RIS configuration.
Specifically, in Section~\ref{sec:LocFine}, we use a gradient-based quasi-Newton method to fuse the coarse position estimations $\{ \hat{\bp}_\text{ce}^\ell \}_{\ell=1}^L$ into a refined  $\hat{\bp}$ with higher localization accuracy.
The coarse and fine localization modules together form the proposed two-stage coarse-to-fine localization algorithm in Section~\ref{sec:LocAlgo}.

\section{Balanced Signaling and Localization Problem Formulation} \label{sec: problem}

In this section, we first present the balanced signaling method and formulate the RIS-assisted localization problem.
Subsequently, we derive the CRLB for each position $\bp\in\cM$ and highlight the connection between RIS beamforming gain and this localization error bound.

\subsection{Balanced Signaling Method} \label{subsec: balanced_method}
Note that the received signals $\{\by_{n,t}\}_{n,t}$ in \eqref{eqn: received_signal} are comprised of the RIS reflected signals $\sum_{\ell=1}^L  \bH^\text{R,M}_{n,\ell}\boldsymbol \Omega_{t,\ell} \bH^\text{M,T}_{n,\ell}\bx_{n,t} $ and non-RIS reflected signals $\bH_n^\text{NLoS} \bx_{n,t} $.
The RIS-reflected signal is adjusted by the RIS reflection coefficients $\boldsymbol \Omega_{t,\ell}$ at different time $t$.
The  balanced signaling method  exploits the time-variant $\boldsymbol \Omega_{t,\ell}$ to
separate the RIS-reflected LoS signal from the $\{\by_{n,t}\}_{n,t}$ in \eqref{eqn: received_signal} with the following procedures.

We set the RIS reflection coefficients $\boldsymbol \Omega_{t,\ell}$ to satisfy the following balanced sequence condition in time domain
\begin{equation}\small
\sum_{t=1}^{T}\boldsymbol \Omega_{t,\ell} = \mathbf{0}_{\frac{M}{L} \times \frac{M}{L}},\;\;\ell= 1,2,\ldots,L.
\label{eqn: RIS_phase_set}
\end{equation}
It guarantees $\sum_{t=1}^{T} \sum_{\ell=1}^L  \bH^\text{R,M}_{n,\ell}\boldsymbol \Omega_{t,\ell} \bH^\text{M,T}_{n,\ell} = \boldsymbol{0}_{N_\text{R} \times N_\text{T}} $, and thus the measurements of NLoS components $\bH_n^\text{NLoS} \bx_{n,t} $ are separated by averaging the measurements $\{ \by_{n,t} \}_{t=1}^T$
with respect to the time $t$, which is 
\begin{equation}
\small \frac{1}{T}\sum_{t=1}^T \by_{n,t}\!=\!\bH_n^\text{NLoS} \bx_{n,t}+ \frac{1}{T}\sum_{t=1}^T \bw_{n,t},\;\; n\!=\!1,2,\ldots,N.
\label{eqn: LOS_separated}
\end{equation}
On this basis, the measurements of the RIS-reflected LoS components are extracted as
\begin{equation}
\small \bar{\by}_{n,t} :=\! \by_{n,t} \!-\! \frac{1}{T}\sum_{t=1}^T \by_{n,t} = \underbrace{ \sum_{\ell=1}^L  \bH^\text{R,M}_{n,\ell}\boldsymbol \Omega_{t,\ell} \bH^\text{M,T}_{n,\ell}\bx_{n,t}}_{ \tilde{\by}_{n,t} (\bp) }  + \tilde{\bw}_{n,t}, 
\label{eqn: ris_extracted}
\end{equation}
where the last equality is obtained by substituting the  \eqref{eqn: LOS_separated} into \eqref{eqn: received_signal}. 
The extracted RIS-reflected component in \eqref{eqn: ris_extracted} is denoted by $\bar{\by}_{n,t}$ while its noise-free version is denoted by $\tilde{\by}_{n,t}(\bp)$.
The noise $\tilde{\bw}_{n,t}$ in the measurements $\{ \bar{\by}_{n,t} \}_{n,t}$ is given by
\begin{equation}
\small
\tilde{\bw}_{n,t} = \bw_{n,t} \- \frac{1}{T}\sum_{t=1}^T \bw_{n,t}. \label{eqn: co_noise}
\end{equation}
For ease of exposition, we denote the collection of all measurements $\{\bar{\by}_{n,t}\}_{n,t}$ in \eqref{eqn: ris_extracted} by a three-order measurement tensor $\bar{\dY} \!\in\! \mathbb{C}^{N \times T \times N_\text{R}}$ such that $\bar{\dY}_{n,t,:} \= \bar{\by}_{n,t}$.
Similar definitions are applicable to $\tilde{\dY} \!\in\! \mathbb{C}^{N \times T \times N_\text{R}}$, $\tilde{\dW} \!\in\! \mathbb{C}^{N \times T \times N_\text{R}}$, and ${\dW} \!\in\! \mathbb{C}^{N \times T \times N_\text{R}}$ such that $\tilde{\dY}_{n,t,:} \= \tilde{\by}_{n,t}$, $\tilde{\dW}_{n,t,:} \= \tilde{\bw}_{n,t}$, and ${\dW}_{n,t,:} \= {\bw}_{n,t}$.
Therefore, the separated measurements $\{\bar{\by}_{n,t}\}_{n,t}$ in \eqref{eqn: ris_extracted} are rewritten as 
\begin{equation}
\small
\bar{\dY} = \tilde{\dY}(\bp) + \tilde{\dW}.
\label{eqn: extracted_model}
\end{equation}

\begin{remark}\label{remark:model_motivation_loc}
Since there is only the LoS path in the cascaded RIS-reflected channel $\bH^\text{R,M}_{n,\ell}\boldsymbol \Omega_{t,\ell} \bH^\text{M,T}_{n,\ell}$
and positions of both the BS and RIS and the rules of setting $\{ \boldsymbol \Omega_{t,\ell}\}_{t,\ell}$ in \eqref{eqn: ris_extracted} are known by the UE, the channel parameters $\{\rho_\text{R,U}^l, \tau^l_\text{R,U}, \theta_\text{R,U}^l, \phi^l_\text{R,U}\}_{l=1}^L$ are functions of the UE position $\bp$.
Therefore, we can treat the $\tilde{\by}_{n,t}$ in \eqref{eqn: ris_extracted} and $\tilde{\dY}$ in \eqref{eqn: extracted_model} as the function of UE position $\bp$, i.e., $\tilde{\by}_{n,t}= \tilde{\by}_{n,t}(\bp)$ and $\tilde{\dY} = \tilde{\dY}(\bp)$.
This enables us to use only the RIS-reflected component $\{\bar{\by}_{n,t}\}_{n,t}$ or $\bar{\dY}$ to formulate the UE localization problem.

\end{remark}

\subsection{Localization Problem Formulation}  \label{subsec: loc_problem}
We now formulate the localization problem using the separated measurements $\bar{\dY}$ in \eqref{eqn: extracted_model}.
Denote by $\hat{\bp}$ the maximum likelihood estimate for the UE's position $\bp$.
The localization problem is expressed as the following maximum likelihood estimation problem:
\begin{equation}
\small
\hat{\bp}  = \argmax_{\bp\in \cM} \log\; p\left( \bar{\dY} | \bp \right) ,
\label{prob: mle1}
\end{equation}
where $\cM$ is the area of interest.
Given the UE's position $\bp$, $p\left(\bar{\dY}|\bp  \right)$ is the probability density function with respect to the observation $\bar{\dY}$ in \eqref{eqn: extracted_model}.

Note that, in the observation model \eqref{eqn: extracted_model}, the covariance matrix of Gaussian noise $\tilde{\dW}_{n,:,i}$ for each index $n$ and $i$ is $\mathrm{Cov}(\tilde{\dW}_{n,:,i})\= \sigma^2 \bI_T \- {\sigma^2}\boldsymbol{1}_T \boldsymbol{1}_T^\rt/T$
due to the balanced signaling procedure.
Hence, the log-likelihood function in \eqref{prob: mle1} has a closed form given by 
\begin{multline}
\footnotesize \log\; p\left( \bar{\dY} | \bp  \right) \propto
\footnotesize -\sum_{n=1}^{N} \sum_{i=1}^{N_\text{R}} 
\left(  \tilde{\dY}_{n,:,i} (\bp)   -\bar{\dY}_{n,:,i}  \right) ^\rh  
\\\footnotesize \times \left( \bI_T - \frac{1}{T} \boldsymbol{1}_T \boldsymbol{1}_T^\rt\right) 
\left( \tilde{\dY}_{n,:,i} (\bp)  - \bar{\dY}_{n,:,i}  \right),
\label{eqn: log-likelihood}
\end{multline}
Therefore, the localization problem in \eqref{prob: mle1} is equivalent to
\begin{align}
&\footnotesize \hat{\bp}  = \argmin_{\bp \in \cM} \sum_{n=1}^{N} \sum_{i=1}^{N_\text{R}} \left\lbrace
\left\|  \tilde{\dY}_{n,:,i} (\bp)   - \bar{\dY}_{n,:,i}  \right\|^2_2 
\right.  \notag\\ 
 &\footnotesize \left. ~~~~~~~~ - \frac{1}{T} \left\| \boldsymbol{1}_T^\rt \tilde{\dY}_{n,:,i} (\bp)- \boldsymbol{1}_T^\rt\bar{\dY}_{n,:,i}  \right\|^2_2 \right\rbrace
\overset{(a)}{=} \argmin_{\bp \in \cM}  \left\|  \tilde{\dY}(\bp)-\bar{\dY}  \right\|_F^2,
\label{prob: mle2}
\end{align}
where the equality (a) holds from the balanced configuration of the RIS coefficients in \eqref{eqn: RIS_phase_set}, i.e., $\boldsymbol{1}_T^\rt \tilde{\dY}_{n,:,i} (\bp)\=0$.

\subsection{Localization Error Bound } \label{subsec: CRLB}
In this subsection, we derive the CRLB and show that this localization error bound is upper bounded by the inverse of RIS beamforming gain.

For any unbiased estimate $\hat{\bp}$ of a position ${\bp}$, there exists 
a theoretically-attainable CRLB for the root-mean-square error (RMSE)  $\sqrt{\mathbb{E} \| \hat{\bp}   -{\bp} \|^2_2}$, which is denoted as $\text{CRLB}\left( \bp \right)$ at a position  $\bp$.
Using the likelihood function in \eqref{eqn: log-likelihood}, 
the element of Fisher information matrix (FIM) $[\bI(\bp)]_{j,k}$ is given by the direct result of computing the FIM of Gaussian distribution \cite{elzanatyReconfigurableIntelligentSurfaces2021}, which is,
\begin{multline}
\footnotesize [\bI(\bp)]_{j,k} \=\frac{2}{\sigma^2}\sum_{n=1}^{N}\sum_{i=1}^{N_\text{R}}
\real \left\lbrace \frac{\partial^\rh  \tilde{\dY}_{n,:,i}  }{\partial p_j}  \left( \bI_T \- \frac{1}{T} \boldsymbol{1}_T \boldsymbol{1}_T^\rt\right) 
\frac{\partial \tilde{\dY}_{n,:,i}  }{\partial p_k} 
\right\rbrace 
\\ \footnotesize \overset{(a)}{=}  \frac{2}{\sigma^2}   \sum_{n,t,i} \real \left\lbrace  \left(  \frac{\partial \tilde{\dY}_{n,t,i} }{\partial p_j}  \right)^*  \frac{\partial \tilde{\dY}_{n,t,i} }{\partial p_k} \right\rbrace,
\label{eqn:FIM_2}
\end{multline}
where $j,k\in \{1,2\}$.
The last equality (a) holds because $\boldsymbol{1}_T^\rt \tilde{\dY}_{n,:,i}$ is constant  with respect to $\bp$ and has zero gradients, which is guaranteed by the balanced configuration of RIS in \eqref{eqn: RIS_phase_set}.
With the FIM, the CRLB at a position $\bp$ is given by
\begin{equation}\small
\sqrt{\mathbb{E} \left\|\hat{\bp} - \bp \right\|_2^2 }\ge \sqrt{\tr\left\lbrace \bI^{-1}\left(\bp \right) \right\rbrace }\treq \mathrm{CRLB}\left( \bp \right).
\label{crlb:lower_bound_pos}
\end{equation}

The closed-form expression of the $\text{CRLB}\left( \bp \right)$ in \eqref{crlb:lower_bound_pos} is often complicated \cite{JiangGao2022RIS}. 
Previous work \cite{heLargeIntelligentSurface2019,bjornsonReconfigurableIntelligentSurfaces2022} has found through numerical simulations that increasing the RIS beamforming gain can reduce the localization error.
The following proposition~\ref{proposition: crlb_and_powergian} theoretically explains this phenomenon and reveals the relationship between the RIS beamforming gain and CRLB.

In proposition~\ref{proposition: crlb_and_powergian}, we show that the $\text{CRLB}( \bp)$ is upper bounded by the inverse of the RIS beamforming gain $F(\boldsymbol{\psi}_t)$. 
Here, the RIS beamforming gain $F(\boldsymbol{\psi}_t)$, as defined following the convention \cite{heLargeIntelligentSurface2019}, is introduced 
by substituting the LoS channel expressions in \eqref{eqn: channel_bs_ris} and \eqref{eqn: channel_ris_ue} into the signal model  \eqref{eqn: ris_extracted}. 
The noise-free received signal power $\Vert  \tilde{\dY}_{n,t,:}  \Vert^2_2$ satisfies the inequality
\begin{equation*}\footnotesize 
\begin{aligned}
\left\|  \tilde{\dY}_{n,t,:}  \right\|^2_2  \=  \Big\|  \sum_{\ell = 1}^L  \left( \boldsymbol{\psi}_{t,\ell}^\rt \bg_{\ell}  \right) \bz_{n,\ell}\Big\|^2_2 \!\le\! \underbrace{ \sum_{\ell = 1}^L   \left\| \boldsymbol{\psi}_{t,\ell}^\rt \bg_{\ell}    \right\|_2^2}_{F(\boldsymbol{\psi}_{t})} \!\times\! \sum_{\ell = 1}^L   \left\| {\bz}_{n,\ell}  \right\|_2^2,
\end{aligned}
\end{equation*}
where $\boldsymbol{\psi}_t \= [\boldsymbol{\psi}_{t,1}^\rt,\boldsymbol{\psi}_{t,2}^\rt,\ldots,\boldsymbol{\psi}_{t,L}^\rt]^\rt \!\in\! \mathcal{F}^{M}$ and $\boldsymbol{\psi}_{t,\ell} \!=\! \diag(\boldsymbol  \Omega_{t,\ell})\in \mathcal{F}^{M/L}$ represents the RIS phase of the $\ell$-th RIS partition.
The term $F(\boldsymbol{\psi}_{t}) \treq \sum_{\ell = 1}^L  \| \boldsymbol{\psi}_{t,\ell}^\rt \bg_{\ell}  \|_2^2$ defines the RIS beamforming gain, and $\bg_\ell \= \ba^*_{{M}/{L}} ( \cos \theta^\text{R,M}_\ell) \odot \ba_{{M}/{L}}(  \cos \phi_\ell^\text{M,T})$ is the end-to-end channel response. 
Both $\bg_\ell$  and 
{\small \begin{multline}
\bz_{n,\ell} = \sqrt{\rho^\text{M,T}_\ell \rho^\text{R,M}_\ell  } \exp\left\lbrace -j 2\pi \left(\tau_\ell^\text{M,T} + \tau_\ell^\text{R,M} \right) f_n \right\rbrace  \\ 
\times \ba_{N_\text{R}}\left( \cos \phi_\ell^\text{R,M}\right) \ba^\rh_{N_\text{T}} \left( \cos \theta^\text{M,T}_\ell\right) \bx,
\label{eqn: z_n}
\end{multline}
}%
are functions of UE position $\bp$ and are independent of the RIS coefficient $\{ \boldsymbol{\psi}_{t,\ell}\}_{t,\ell}$.

\begin{proposition}
Suppose the UE position  $\bp\=[p_1,p_2]^\rt$ is in a compact area.
If the following two constraints
\begin{equation}
\small \left\| \boldsymbol{\psi}_{t,\ell}^\rt \bg_{\ell}    \right\|_2 \neq 0, \quad \left\|  \frac{\partial \bz_{n,\ell} }{\partial  p_d} \right\|_2  \neq 0,\forall t,\ell,n,d,
\label{eqn: crlb_power_mild}
\end{equation}
hold, we have
\begin{align}
\small \mathrm{CRLB}\left( \bp\right) \le C_0\left[\sum_{t=1}^{T} \sum_{\ell = 1}^L   \left\| \boldsymbol{\psi}_{t,\ell}^\rt \bg_{\ell}    \right\|_2^2 \right]^{-\frac{1}{2}} \!\!\!\!\!= C_0\!\left[\sum_{t=1}^{T} F(\bpsi_t) \right]^{-\frac{1}{2}} ,
\label{eqn: crlb_lowerbound}
\end{align}
where $C_0\!>\!0$ is a constant and independent of the RIS coefficient $\{ \boldsymbol{\psi}_{t,\ell}\}_{t,\ell}$ and the UE position $\bp$.
\label{proposition: crlb_and_powergian}
\end{proposition}

\begin{proof}
Please see Appendix \ref{appendix-C}.
\end{proof}

In  Proposition~\ref{proposition: crlb_and_powergian}, the constraint \eqref{eqn: crlb_power_mild}  requires that the reflection coefficients of the $\ell$-th RIS partition $\boldsymbol{\psi}_{t,\ell}$ are configured to be nonorthogonal with the combined end-to-end channel response  $\bg_\ell$, and that the $\bz_{n,\ell}$ is not completely insensitive to changes of the UE's position $\bp$.
These requirements are typically satisfied in most cases.
Additionally, the upper bound in \eqref{eqn: crlb_lowerbound} motivates us to maximize the passive beamforming gain $F(\boldsymbol{\psi}_t)$ to lower the $\text{CRLB}\left( \bp \right)$.

\section{Optimal Algorithm for Discrete Beamforming Problem} \label{sec: FIM}
In this section, we formulate the discrete beamforming problem to design the RIS reflection coefficient for achieving a high RIS beamforming gain. 
This is a combinatorial optimization problem and can be optimally solved by the proposed  FPB algorithm in linear time.

\subsection{Discrete Beamforming Problem} \label{subsec: beamform_problem}

According to \eqref{eqn: crlb_lowerbound}, we can 
improve the localization performance by maximizing the RIS beamforming gain $F(\boldsymbol{\psi}_t)$, which yields a discrete beamforming problem.

The RIS beamforming gain $F(\boldsymbol{\psi}_{t})$ in \eqref{eqn: crlb_lowerbound} is a bivariate function of the RIS reflection coefficient $\{\boldsymbol  \Omega_{t,\ell}\}_{t,\ell}$ and the unknown UE position $\bp$. 
In the first $T/2$ time slots, i.e., $t\=1,2,\ldots, T/2$, there is no prior to the UE’s position $\bp$ and we collect  measurements for 
the coarse localization model.
Here,  we  set the 
RIS reflection coefficient $\{\boldsymbol  \Omega_{t,\ell}\}_{t,\ell}$ as random or discrete Fourier transform-based phases \cite{dardariNLOSNearFieldLocalization2021}, which  satisfies the balanced condition $\sum_{t=1}^{T/2}\boldsymbol \Omega_{t,\ell} = \mathbf{0}_{\frac{M}{L} \times \frac{M}{L}}$.
For time slots $t\ge T/2+1$, we fix the UE position $\bp$ as the coarse localization result, as demonstrated in Fig.~\ref{fig: task}, and
design the RIS reflection coefficient $\{\boldsymbol  \Omega_{t,\ell}\}_{t,\ell}$ to maximizing the RIS beamforming gain $F(\boldsymbol{\psi}_{t})$.

As mentioned in Section~\ref{sec: protocol}, the coarse localizaton model provides $L$ position estimation in $\{\hat{\bp}_\text{ce}^\ell \}_{\ell=1}^L$. 
A high-quality estimate of the $\bp$ is the one in $\{\hat{\bp}_\text{ce}^\ell \}_{\ell=1}^L$ yielding the minimum objective value in \eqref{prob: mle2}, and we denote it by $\hat{\bp}_\text{ce}$.
Given the coarse estimate $\hat{\bp}_\text{ce}$, the $\bg_l$ and $\bz_{n,l}$ are known, and the RIS beamforming gain $ F(\boldsymbol{\psi}_t)$ becomes a univariate function of the RIS coefficient  $\boldsymbol{\psi}_t = [\boldsymbol{\psi}_{t,1}^\rt,\boldsymbol{\psi}_{t,2}^\rt,\ldots,\boldsymbol{\psi}_{t,L}^\rt]^\rt$.

Besides, note that the $\bg_\ell$ in $F(\boldsymbol{\psi}_t) \= \sum_{\ell = 1}^L \| \boldsymbol{\psi}_{t,\ell}^\rt \bg_{\ell} \|_2^2$ is independent of time $t$, so we only need to design the optimal RIS reflection coefficient for the time slot $t\=T/2+1$.
The reflection coefficients at  $t\=T/2+2, T/2 + 3,\ldots, T$ can be easily derived from those at $t\=T/2+1$ to satisfy the balanced condition \eqref{eqn: RIS_phase_set} from the time slot $t\=T/2+1$ to $T$. 
Whenever there is no ambiguity, in what follows, we discuss the case at $t\=T/2+1$ and simply denote the RIS coefficient matrix
$\boldsymbol{\psi}_{T/2 + 1,\ell}$ as  $\boldsymbol{\psi}_{\ell}$.

Therefore,  the discrete beamforming problem for the RIS reflection coefficient design is formulated as
\begin{subequations}
    \begin{equation}\small
        \underset{ \boldsymbol{\psi}  }{\text{maximize}}  \;\; \sum_{\ell = 1}^L \left\|  \boldsymbol{\psi}_{\ell} ^\rt \bg_\ell \right\|^2_2 \treq F(\boldsymbol{\psi})
        \label{eqn: beamforming_problem_obj} 
        ~~~~~~~~~~~~~~~~~~~~~~~~~~
    \end{equation}
    \begin{equation}\small
        \text{s.t.}~ \psi_{\ell,k}\!\in\!\! \left\lbrace \tiny\exp\lp j\frac{2\pi}{2^b}s \rp,s \= 0,1,\ldots,2^b - 1 \right\rbrace,\;\forall \ell,k, 
\label{Eqn:Constraint}
    \end{equation}
    \label{eqn: beamforming_problem} 
\end{subequations}where $\boldsymbol{\psi} \= [\boldsymbol{\psi}_{1}^\rt,\boldsymbol{\psi}_{2}^\rt,\ldots,\boldsymbol{\psi}_{L}^\rt]^\rt \!\in\! \mathcal{F}^{M}$. Without the discrete constraint \eqref{Eqn:Constraint} on each element of $\boldsymbol{\psi}$, i.e., $\psi_{\ell,k}$, problem \eqref{eqn: beamforming_problem}  will be degraded into a maximum-ratio-combination problem and the optimal solution is $\boldsymbol{\psi}^\star_{\ell} \= \bg_\ell^*, \forall \ell$. With the 
constraint \eqref{Eqn:Constraint}, 
the problem \eqref{eqn: beamforming_problem} is indeed a combinatorial optimization programming, which requires exponential complexity to obtain the optimal solution. 
Some widely-used polynomial-time approaches, such as the Closest Point Projection (CPP) method \cite{wuBeamformingOptimizationWireless2020,youChannelEstimationPassive2020} and the Semidefinite Relaxation (SDR) method\cite{wuIntelligentReflectingSurface2019b,linReconfigurableIntelligentSurfaces2021}, can only provide a sub-optimal solution.
In the next subsection, we will propose a linear-time FPB algorithm to optimally solve the combinatorial optimization problem \eqref{eqn: beamforming_problem}.

\subsection{Fast Passive Beamforming (FPB) Algorithm} \label{subsec: FPB}

We begin with providing a  necessary condition for the optimal solution $\bpsi_\ell^\star \!\in\! \mathcal{F}^{M/L}$ of problem \eqref{eqn: beamforming_problem} in Lemma \ref{lemma: opt_condition}.
This necessary condition serves as a geometric characterization for the optimal solution $\bpsi_\ell^\star$, indicating that each RIS phase shift $\bpsi_{\ell,k}^\star$ should be adjusted to make the angle between $\psi^\star_{\ell,k_1} g_{\ell,k_1} $ and $ \psi^\star_{\ell,k_2} g_{\ell,k_2} $ as close as possible for $k_1 \neq k_2$ where $\psi^\star_{\ell,k} $ and $g_{\ell,k} $ are the $k$-th elements of $\boldsymbol{\psi}_\ell^\star$ and $\bg_{\ell}$, respectively.

\begin{lemma}
The optimal solution $\{\boldsymbol{\psi}_\ell^\star\}_{\ell=1}^L$ of problem \eqref{eqn: beamforming_problem} must satisfy 
\begin{equation}\small
 A\bigg( \sum_{k \in S_1}\psi^\star_{\ell,k} g_{\ell,k} ,\sum_{k \in S_2} \psi^\star_{\ell,k} g_{\ell,k} \bigg) \le \frac{\pi}{2^b},\,\forall \, \ell,
\label{eqn: opt_condition}
\end{equation}
where $A(x,y) \!=\!  \pi \!-\!  \left|   \left| \mathrm{arg}(x) \!-\! \mathrm{arg}(y) \right| \!-\! \pi \right|$ is the function that calculates the internal angle of complex variables $x,y \!\in\! \C$.
In \eqref{eqn: opt_condition}, $\{S_1,S_2\}$ is any partition of the set $\{1,2,\ldots,M/L\}$, i.e., $S_1 \!\cap\! S_2 \= \emptyset, \; S_1 \!\cup\! S_2 \= \{1,2,\ldots,M/L\}$.
\label{lemma: opt_condition}
\end{lemma}
\begin{proof}
Please see Appendix \ref{appendix-D}.
\end{proof}

According to Lemma~\ref{lemma: opt_condition}, the optimal solution $\{\boldsymbol{\psi}^\star_\ell\}_{\ell=1}^L$ of problem \eqref{eqn: beamforming_problem} is in the set of solutions   satisfying the condition in \eqref{eqn: opt_condition}.
In the following Lemma~\ref{lemma: equivalence_omega_beamformer}, we find the closed-form expression of $\{\boldsymbol{\psi}_\ell\}_{\ell=1}^L$ satisfying the  inequality \eqref{eqn: opt_condition}.

\begin{lemma}
A closed-form expression of 
$\{\boldsymbol{\psi}_\ell\}_{\ell=1}^L$ satisfying the  inequality constraint \eqref{eqn: opt_condition}
is given by 
\begin{equation}\small
\boldsymbol \psi_{\ell} 
= \exp \left\lbrace j\frac{2\pi}{2^b} \mathrm{round} \left(-\frac{2^b}{2\pi}\mathrm{arg}(\bg_\ell) + \frac{2^b}{2\pi}\omega_\ell  \right) \right\rbrace \label{omega_beamformer}, 
\end{equation}
where $\omega_\ell \= \mathrm{arg}( (\boldsymbol \psi_\ell)^\rt \bg_\ell)$.
The function $\mathrm{round}(x)$ returns the nearest integer to  $x$.
\label{lemma: equivalence_omega_beamformer}
\end{lemma}
\begin{proof}
Please see Appendix \ref{appendix-E}.
\end{proof}

Note that the $\psi_{\ell}$ in \eqref{omega_beamformer} is a function of the  $\omega_l$. 
For any $\omega \in [0,2\pi)$, we define the $\omega$-beamformer
for the $\ell$-th RIS partition as 
\begin{equation}\small
 {\boldsymbol \psi}_{\ell}(\omega)
= \exp \left\lbrace j\frac{2\pi}{2^b} \mathrm{round} \left(-\frac{2^b}{2\pi}\mathrm{arg}(\bg_\ell) + \frac{2^b}{2\pi}\omega  \right) \right\rbrace \label{omega_beamformer2}.
\end{equation}
It is worth noting that taking $\omega \= 0$ in \eqref{omega_beamformer2}, i.e., ${\boldsymbol \psi}_\ell(0)$, is equivalent to designing the RIS coefficients by the CPP method.
Lemma~\ref{lemma: opt_condition} and Lemma~\ref{lemma: equivalence_omega_beamformer} reveal that the  optimal solution $\{\boldsymbol{\psi}_\ell^\star\}_{\ell=1}^L$ is contained in the $\omega$-beamformer set $\{{\boldsymbol \psi}_{\ell}(\omega),\omega \!\in\! [0,2\pi)\}$. 
This means we can search over the set  $\{{\boldsymbol \psi}_{\ell}(\omega), \omega \in [0,2\pi)\}$ to obtain the optimal solution of problem \eqref{eqn: beamforming_problem}.
However, due to 
the continuous $\omega \!\in\! [0,2\pi)$, the search involves an infinite number of trials and thus has high complexity. 
The following Proposition \ref{proposition: linear_complexity} indicates that one can conduct the search in linear time with respect to the number of RIS units $M$.

\begin{proposition}
Let $\omega_\ell^\star$ be the optimal solution to the  problem
\begin{equation}\small
\omega_\ell^\star = \argmax_{\omega \in [0,2\pi) } \left\| { \boldsymbol \psi}_{\ell}^\rt(\omega) \bg_\ell \right\|^2_2,\, l=1,2,\ldots,L.
\label{prob: omega_beamforming}
\end{equation}
where ${\boldsymbol \psi}_{\ell}(\omega)$ is the $\ell$-th $\omega$-beamformer as defined in \eqref{omega_beamformer2}.
Then, $\{{ \boldsymbol \psi}_\ell(\omega^\star_\ell) \}_{\ell=1}^L$ is  optimal to problem \eqref{eqn: beamforming_problem}. 
Furthermore, problem \eqref{prob: omega_beamforming} can be optimally solved by searching  over a finite set of order $\mathcal{O}(M)$, that is,
\begin{equation}\small
\omega_\ell^\star = \argmax_{\omega \in \mathcal{F}_\ell } \left\|{ \boldsymbol \psi}_{\ell}^\rt(\omega) \bg_\ell\right\|_2^2,
\label{prob: omega_beamforming2}
\end{equation}
where $\mathcal{F}_\ell$ is defined as
\begin{equation}\small
\begin{aligned}
\mathcal{F}_\ell = &  \left\lbrace 
\mathrm{arg}(g_{\ell,k}) + \frac{2\pi}{2^b}\left(  \left \lceil -\frac{2^b \cdot \mathrm{arg}(g_{\ell,k})}{2 \pi} \- \frac{1}{2}   \right \rceil  \!+\!  s\!+\! \frac{1}{2}\right) , \right.  \\ & ~~~~~~~~~~~~
\left. s=0,1,2,\ldots,2^b-1,k=1,2,\ldots,\frac{M}{L}\right\rbrace,
\end{aligned}\label{eqn: finite_set}
\end{equation}
and $\left \lceil \cdot \right \rceil$ is the ceiling operator.
\label{proposition: linear_complexity}
\end{proposition}
\begin{proof}
Please see Appendix \ref{appendix-F}.
\end{proof}

Proposition~\ref{proposition: linear_complexity} means the problem \eqref{eqn: beamforming_problem}, \eqref{prob: omega_beamforming} and \eqref{prob: omega_beamforming2} are equivalent, and we can  solve the problem \eqref{prob: omega_beamforming2} by searching the $\omega\in\cF_l$  with at most $\sum_{\ell=1}^L \left| \mathcal{F}_\ell \right| \!=\! 2^b \times \frac{M}{L} \times L \!=\! 2^b M$  times. 
This implies the discrete beamforming problem \eqref{eqn: beamforming_problem} can be optimally solved with linear search time $\mathcal{O}(M)$.
In each search, one needs to compute the objective value in \eqref{prob: omega_beamforming2} with complexity $\cO(M)$, and thus the overall computational complexity of the FPB algorithm is quadratic, i.e., $\mathcal{O}(M^2)$.

Based on the designed RIS phase at $t\=T/2 + 1$, we now find the optimal RIS phase configuration at $t\=T/2+2,\ldots, T$. 
Although the 
optimal RIS reflection coefficients $\{\boldsymbol{\psi}_\ell^\star\}_{\ell=1}^L$ at $t\=T/2+1$ can be directly used at  $t\=T/2 +2,T/2 + 3,\ldots, T$ to achieve the maximum passive beamforming gain, the RIS reflection coefficients at  the $t=1,2,\ldots,T$ should be designed to 
satisfy the  balanced sequence condition \eqref{eqn: RIS_phase_set}, i.e., $\sum_{t=T/2+1}^T\boldsymbol{\psi}_{t,\ell} \=  \mathbf{0}_{M/L}$. 
Note that there exists a rotation-invariant property among the optimal solutions.
For example, suppose $\boldsymbol{\psi}^\star_t \= [(\boldsymbol{\psi}_{t,1}^\star)^\rt,(\boldsymbol{\psi}_{t,2}^\star)^\rt,\ldots,(\boldsymbol{\psi}_{t,L}^\star)^\rt]^\rt\in \mathcal{F}^{M}$ is the optimal solution of problem \eqref{eqn: beamforming_problem}, then $\exp(j\frac{2\pi s }{2^b}) \boldsymbol{\psi}^\star_t$ for $s \!\in\! \{ 0,1,\ldots,2^b - 1 \}$ turns out to be another optimal solution that shares the same passive beamforming gain, i.e., $F(\boldsymbol{\psi}^\star_t) \= F(\exp(j\frac{2\pi s_\ell }{2^b}) \boldsymbol{\psi}^\star_t)$.
On this basis, we set the RIS reflection coefficients 
at $t \= T/2+1,T/2+2,\ldots, T$ be $\boldsymbol{\psi}_{T/2+1,\ell} \= \boldsymbol{\psi}_{\ell}^\star$ and $\boldsymbol{\psi}_{T/2+t,\ell} \= \exp(j\frac{(t-1) \pi  }{2^{b-1}})\boldsymbol{\psi}_{\ell}^\star$.
On one hand, the optimality of $\{\boldsymbol{\psi}_{t,\ell}\}_{t,\ell}$ is guaranteed due to the rotation-invariant property.
On the other hand, we can verify that $\sum_{t=T/2+1}^T \exp(j\frac{(t-1) \pi  }{2^{b-1}})\boldsymbol{\psi}_{\ell}^\star \= \mathbf{0}_{M/L}$ because $\sum_{t=T/2+1}^T \exp(j\frac{(t-1) \pi  }{2^{b-1}})\= 0$ holds when $T$ is set as a multiple of $2^{b+1}$.
Under this setting, the balanced sequence condition $\sum_{t=1}^T\boldsymbol{\psi}_{t,\ell} \=  \mathbf{0}_{M/L}$ is naturally satisfied.

The result of Proposition \ref{proposition: linear_complexity} and its extension for $t>T/2 + 1$ together form the proposed FPB algorithm, as summarized in Algorithm \ref{alg: FPB_algorithm}.

\begin{algorithm}
\caption{FPB Algorithm}
\label{alg: FPB_algorithm}
\LinesNumbered 
\KwIn{Position of BS $\bp_{\text{T}}$, position of the $\ell$-th RIS segment $\bp_\text{M}^{\ell} (\ell\!=\! 1,2,\ldots,L)$, coarse estimate of UE position $\hat{\bp}_\text{ce}$}
\KwOut{Optimal RIS reflection coefficients $\{ \boldsymbol \psi^\star_{t,\ell} \}_{t,\ell}$ for $t> T/2$}
\For{$\ell \in \{1,2,\ldots,L\}$}{
	Compute the incident AoA $\phi_\ell^\text{M,T}$ from $\bp_{\text{T}}$ to $\bp_\text{M}^{\ell}$\;
	Compute the incident AoD $\theta^\text{R,M}_\ell$ from $\bp_\text{M}^{\ell}$ to $\hat{\bp}_\text{ce}$ \;
	Compute $\bg_\ell = \ba^*_{\frac{M}{L}} \left( \cos \theta^\text{R,M}_\ell\right) \odot \ba_{\frac{M}{L}}\left(  \cos \phi_\ell^\text{M,T}\right) $\;
	Form $\mathcal{F}_\ell$ using \eqref{eqn: finite_set}\;
	Obtain $\omega_\ell^\star$ by solving the problem \eqref{prob: omega_beamforming2}  \;
    Compute $\boldsymbol \psi^\star_{\ell} = { \boldsymbol\psi}_{\ell}(\omega^\star_\ell)$ using \eqref{omega_beamformer2}\;
}

\For{$t \in \{T/2+1,T/2+2,\ldots,T\}$}{
Compute $\boldsymbol{\psi}_{t,\ell}^\star = \exp\left(j\frac{(t-1) \pi  }{2^{b-1}}\right)\boldsymbol{\psi}_{\ell}^\star,\forall \ell$ \;
}
\end{algorithm}

 \vspace{-3mm}
\section{Coarse-To-Fine Localization Algorithm} \label{sec:LocAlgo}
In this section, we propose a two-stage coarse-to-fine algorithm to solve the localization problem \eqref{prob: mle2}. Note that the passive beamforming module in Section~\ref{sec: FIM} is between the coarse and fine localization modules. In the following two subsections, we sequentially present the  coarse  and  fine localization modules.

\subsection{Coarse Localization Module}\label{sec:LocCoarse}
As problem \eqref{prob: mle2} is non-convex, directly solving it  with a gradient descent approach may yield a sub-optimal solution with significant estimation error. 
To circumvent this challenge, we first leverage the extracted AoAs and ToAs from part of the LoS measurements $\bar{\dY}$ in \eqref{eqn: extracted_model}.
Specifically, we use measurements from the initial $T/2$ time slots to coarsely estimate the UE's position $\bp$.
The AoA $f_{\ell}\!:=\! \cos \phi_\ell^\text{R,M}$ and ToA $\tau_\ell \!:=\! \tau_\ell^\text{M,T} + \tau_\ell^\text{R,M}$ of   the $L$ RIS-partitioned segments  are extracted from the  
partitioned-far-filed representation of the BS-RIS-UE channel in \eqref{eqn: ris_extracted}, i.e., $\sum_{\ell=1}^L  \bH^\text{R,M}_{n,\ell}\boldsymbol \Omega_{t,\ell} \bH^\text{M,T}_{n,\ell}$.
To do this, we substitute the expressions in \eqref{eqn: channel_bs_ris} and \eqref{eqn: channel_ris_ue} into \eqref{eqn: ris_extracted} and obtain 
\begin{equation}\small
\begin{bmatrix}\bar{\dY}_{1,t,:}\\ \bar{\dY}_{2,t,:} \\ \vdots \\ \bar{\dY}_{N,t,:}   \end{bmatrix} =
\sum_{\ell=1}^L c_{t,\ell} \ba_N( 2\Delta f \tau_\ell) \otimes {\ba}_{N_\text{R}}(f_\ell) +  
\begin{bmatrix}\tilde{\dW}_{1,t,:}\\ \tilde{\dW}_{2,t,:}  \\ \vdots \\ \tilde{\dW}_{N,t,:}  \end{bmatrix},
\label{eqn: tensor}
\end{equation}
where  the $\ell$-th element of the amplitude vector $\bc_{t}=[c_{t,1},c_{t,2},\ldots,c_{t,L}]^\rt$ is
$  c_{t,\ell} =\sqrt{\rho^\text{M,T}_\ell \rho^\text{R,M}_\ell  } \exp \lbrace -j 2\pi \tau_\ell  (f_c -\frac{N-1}{2}\Delta f )  \rbrace    ( \boldsymbol{\psi}_{t,\ell}^\rt \bg_{\ell}  )\ba^\rh_{N_\text{T}}  ( \cos \theta^\text{M,T}_\ell ) \bx$ and $t\=1,2,\ldots,T/2$.

Note that estimating the two-dimensional parameters $\{(\tau_\ell,f_{\ell})\}_{\ell=1}^L$ from \eqref{eqn: tensor} is challenging since existing methods such as the two-dimensional multiple signal classification algorithm and the two-dimensional simultaneous orthogonal matching pursuit algorithm often incur extremely high time complexity.
To mitigate this issue, we propose to relax the two-dimensional parameter estimation problem \eqref{eqn: tensor} into two one-dimensional sub-problems, which significantly reduces the computational complexity.

With a little abuse of notation, we stack all observations in \eqref{eqn: tensor} and denote ${\bY}_{t},{\bW}_{t} \in \C^{N_\text{R} \times N}$ as
\begin{equation*}\footnotesize
\begin{aligned}
 {\bY}_{t}  \!=\! \left[ \bar{\dY}_{1,t,:},\bar{\dY}_{2,t,:},\ldots,\bar{\dY}_{N,t,:} \right],
  {\bW}_{t} \!=\! \left[  \tilde{\dW}_{1,t,:},\tilde{\dW}_{2,t,:},\ldots,\tilde{\dW}_{N,t,:}\right] .
\end{aligned}
\end{equation*}
Using the symmetry of the Kronecker product, the \eqref{eqn: tensor} is rewritten into
\begin{equation*}\small
\begin{aligned}
 {\bY}_{t}^\rt & = \left[ \ba_N( 2\Delta f \tau_1),\ldots,\ba_N( 2\Delta f \tau_L)\right]  \bC_{t}^\text{ToA}  +  {\bW}_{t}^\rt, \\
 {\bY}_{t} &= \left[ \ba_{N_\text{R}}(f_{1}), \ldots,\ba_{N_\text{R}}(f_{L})\right]  \bC_{t}^\text{AoA}  +  {\bW}_{t},
\end{aligned}
\end{equation*}
where $\bC_{t}^\text{ToA} \!=\! \diag(\bc_t) \left[ \ba_{N_\text{R}}(f_{1}), \ldots,\ba_{N_\text{R}}(f_{L})\right] ^\rt \!\in\! \C^{L \times N_\text{R}}$ and $\bC_{t}^\text{AoA} \!=\! \diag(\bc_t) \left[ \ba_N( 2\Delta f \tau_1),\ldots,\ba_N( 2\Delta f \tau_L)\right]^\rt \!\in\! \C^{L \times N}$.
Although each of the $L$ RIS-partitioned segments has a ToA and an AoA, in the coarse localization module, we only estimate one ToA and one AoA using the beam scanning philosophy given by 
\begin{equation}\small
\hat{\tau} = \argmax_{\tau \in [0,1/\Delta f]}  \sum_{t=1}^{T/2} \left\| \ba_N^\rh( 2\Delta f \tau)\bar{\bY}_{t}^\rt  \right\|^2_2,
\label{pro:toa_fft}
\end{equation}	
and
\begin{equation}\small
\hat{f} = \argmax_{f \in [-1,1]}  \sum_{t=1}^{T/2} \left\|    {\ba}_{N_\text{R}}^\rh(f)  \bar{\bY}_{t}  \right\|^2_2.
\label{pro:aoa_fft}
\end{equation}
Sub-problems \eqref{pro:toa_fft} and \eqref{pro:aoa_fft} can be efficiently solved using the inverse fast Fourier transform (IFFT) method. 

Next, we utilize the estimated ToA $\hat{\tau}$ and  AoA $\hat{f}$ to determine the UE's position $\bp$ coarsely. Note that the ToA $\hat{\tau}$ and AoA $\hat{f}$ can be related to any of the $L$ RIS segments. 
We assume that the estimated  $(\hat{\tau},\hat{f})$ is associated with the reflection path through each of the $L$ RIS partitions and generate in total $L$ coarse estimates 
 $\{\hat{\bp}_\text{ce}^\ell\}_{\ell=1}^L$  given by
\begin{equation}\small
\hat{\bp}_\text{ce}^\ell = \bp_{\text{M}}^{\ell}  + \left(  c \hat{\tau} - \left\|   \bp_{\text{M}}^{\ell}  - \bp_{T} \right\|_2 \right) \left[ -\hat{f},-\sqrt{ 1 - \hat{f^2}} \right] ^\rt,\, \forall \,l,
\label{eqn: geometric} 
\end{equation}
where $\bp_{\text{M}}^{\ell}$ is the position of the $\ell$-th RIS partition.
As emphasized in the proposed protocol in Section~\ref{sec: protocol}, the information of coarse localization result $\{\hat{\bp}_\text{ce}^\ell \}_{\ell=1}^L$ in \eqref{eqn: geometric} is utilized by both the passive beamforming module and the fine localization module.
The passive beamforming module employs the one in $\{\hat{\bp}_\text{ce}^\ell \}_{\ell=1}^L$ with the minimum objective value in \eqref{prob: mle2}, i.e., $\hat{\bp}_\text{ce}$, as the input.
On the other hand, the fine localization module utilizes all the $L$ coarse estimates $\{\hat{\bp}_\text{ce}^\ell \}_{\ell=1}^L$ to enhance the localization accuracy, as detailed at the below.

\subsection{Fine Localization Module}\label{sec:LocFine}
In this subsection, we exploit the LoS measurements $\bar{\dY}$ at all the $T$ time slots in \eqref{eqn: extracted_model} and propose a gradient-based method to fuse the $L$ coarse estimates $\{ \hat{\bp}_\text{ce}^\ell \}_{\ell=1}^L$ obtained from \eqref{eqn: geometric} into a refined estimate.
Although the $\{ \hat{\bp}_\text{ce}^\ell \}_{\ell=1}^L$ are coarse estimates, their accuracies can be in sub-meter level and 
the accuracy of these coarse estimates guarantees that one of the $\{ \hat{\bp}_\text{ce}^\ell \}_{\ell=1}^L$ is highly likely to be in the vicinity of the global optimum.
Therefore, the basic idea of the proposed fusion method is to apply the gradient-based quasi-Newton iteration initiated with each coarse estimate $\hat{\bp}_\text{ce}^\ell$ independently for $\ell=1,2,\ldots, L$.
Intuitively, one of the quasi-Newton iterations may approach the global optimum through the procedure of gradient descent.
The $\ell$-th iterative scheme is given by
\begin{equation}\small
\bp_{0}^{\ell} =\hat{\bp}_\text{ce}^\ell,\; \bp_{k}^{\ell} = \bp_{k-1}^{\ell}  -\lambda_k \bG_k  \left(\bg_k^{\ell} \right)^\rt,\; k\!=\!1,2,\ldots,K,\label{eqn:itr_grad}
\end{equation}
where $K$ is the maximum number of iterations and $\lambda_k$ is the step size.
The gradient term $\bg_k^\ell\in \R^{1\times D}$ is computed using
\begin{align}
&\footnotesize \bg_k^\ell := \frac{\partial \left\|  \tilde{\dY}\left(\bp_{k-1}^{\ell}\right)-\bar{\dY}  \right\|_F^2 }{\partial \bp}\nonumber \\ 
 &\footnotesize  =  2\sum_{n=1}^{N} \sum_{t=1}^{T} \real\left\lbrace \left(  \tilde{\dY}_{n,t,:}\left(\bp_{k-1}^{\ell} \right) -  \bar{\dY}_{n,t,:} \right)^\rh  \frac{\partial \tilde{\dY}_{n,t,:}\left(\bp_{k-1}^{\ell}\right) }{\partial \bp}\right\rbrace , 
 \label{eqn:y_grad}
\end{align}
where the Jacobi matrix ${\partial \tilde{\dY}_{n,t,:}(\bp) } / {\partial \bp} \in \C^{N_\text{R} \times D}$ is computed using the partitioned-far-file representation in \eqref{eqn: ris_extracted}.
The approximation to the inverse Hessian matrix, $\bG_k$ in \eqref{eqn:itr_grad}, is provided by the well-known Broyden–Fletcher–Goldfarb–Shanno (BFGS) method. 
Among all the refined estimates $\{ \bp_{K}^{\ell} \}_{\ell=1}^L$, the final estimate of the UE's position $\hat{\bp}$ is obtained by selecting the refined estimate with the minimum objective function value in \eqref{prob: mle2}, which is mathematically given by 
\begin{equation}\small
\hat{\bp} = \argmin_{ 1 \le \ell \le L  } \left\| \tilde{\dY}(\bp_{K}^{\ell}) -  \bar{\dY} \right\|^2_F.
\label{eqn:fine_est}
\end{equation}

The coarse-to-fine localization algorithm consists of the coarse localization  and the fine localization modules, as described in Fig.~\ref{fig: task}.
The computational complexity is dominated by the IFFT-based line search in the coarse localization and the gradient-based iteration in the fine localization.
In the stage of coarse localization, the ToA $\hat{\tau}$ and the AoA $\hat{f}$ are estimated by solving \eqref{pro:toa_fft} and \eqref{pro:aoa_fft}.
The complexity is $\mathcal{O}( TN_r N_{F_1}\log(N_{F_1}) )$ and $\mathcal{O}( TN N_{F_2}\log(N_{F_2}))$ where $N_{F_1}$ and $N_{F_2}$ denote the sampling numbers of IFFT in \eqref{pro:toa_fft} and \eqref{pro:aoa_fft}, respectively.
Typically, $N_{F_1} \!\!=\!\! c_0 N$ and $N_{F_2} \!\!=\!\! c_0 N_\text{R}$ where $c_0\!>\!0$ is the oversampling factor.
In the stage of fine localization, the main complexity of each iteration is the computation of the theoretical measurements $\tilde{\dY}(\bp)$ and its gradient ${\partial \tilde{\dY}(\bp) } / {\partial \bp}$ in \eqref{eqn:y_grad}. 
These operations have complexity on the order of $\mathcal{O}( KNTL(N_t M + M N_r))$.

\section{Simulation Studies}   \label{simulation_sec}

In this section, we numerically evaluate the performance of the proposed localization protocol, including the discrete beamforming algorithm and the localization algorithm.

\subsection{Simulation Setup}\label{subsec: case_setting}
We consider a RIS-aided downlink 2D scenario, where the BS is deployed at position $\bp_\text{T} \= (0,0)$m with $N_\text{T}\=32$ transmit antennas, and the UE is randomly located in the area of interest $\cM \= \lp 10,30 \rp \times \lp 10,30\rp \text{m}^2$ with $N_\text{R}\=16$ receive antennas.
The RIS is deployed at position $\bp_\text{M} \= (15,40) {\rm m}$ with the number of RIS elements ranging from $M\=32$ to $M\=256$. To satisfy the criterion of the valid far-field condition in \eqref{eqn: Fraunhofer}, 
the RIS with $M\=256$ should be divided into at least $L \= 4$ 
segments under the partitioned-far-field representation in \eqref{eqn: segment_ris}.
Each RIS element has $2^b$ candidate discrete reflection coefficients, where we set $b\=1$ or $b\=2$ in the subsequent simulation.
We focus on the mmWave system operated at the carrier frequency  $f_c \= 60$Ghz and set the subcarrier bandwidth as $\Delta f\=120$kHz.
During the pilot transmission, we set the numbers of subcarriers $N \= 128$ and transmission time slots $T \= 16$.
It is noteworthy that values of the $N$ and $T$ are set much smaller than those in the existing literature \cite{dardariNLOSNearFieldLocalization2021,huangNearFieldRSSBasedLocalization2022}, which shows a low-overhead advantage of our proposed RIS-assisted localization system.
We set the transmit power $P_\text{T}\= 30$dBm.
The signal-to-noise ratio is defined as $\text{SNR} \=   10 \log_{10} ({P_\text{T}}/{\sigma^2}) $dB, where the noise power $\sigma^2$ varies under different $\text{SNR}$ settings.

\subsection{Simulation Results: Passive Beamforming}\label{subsec: result_bf}

\begin{figure}[!t]
		\centering
	\subfloat[$M=64,b=1$\label{fig: bf_cdf_b1_M64} ]{
		\includegraphics[width=1.6in]{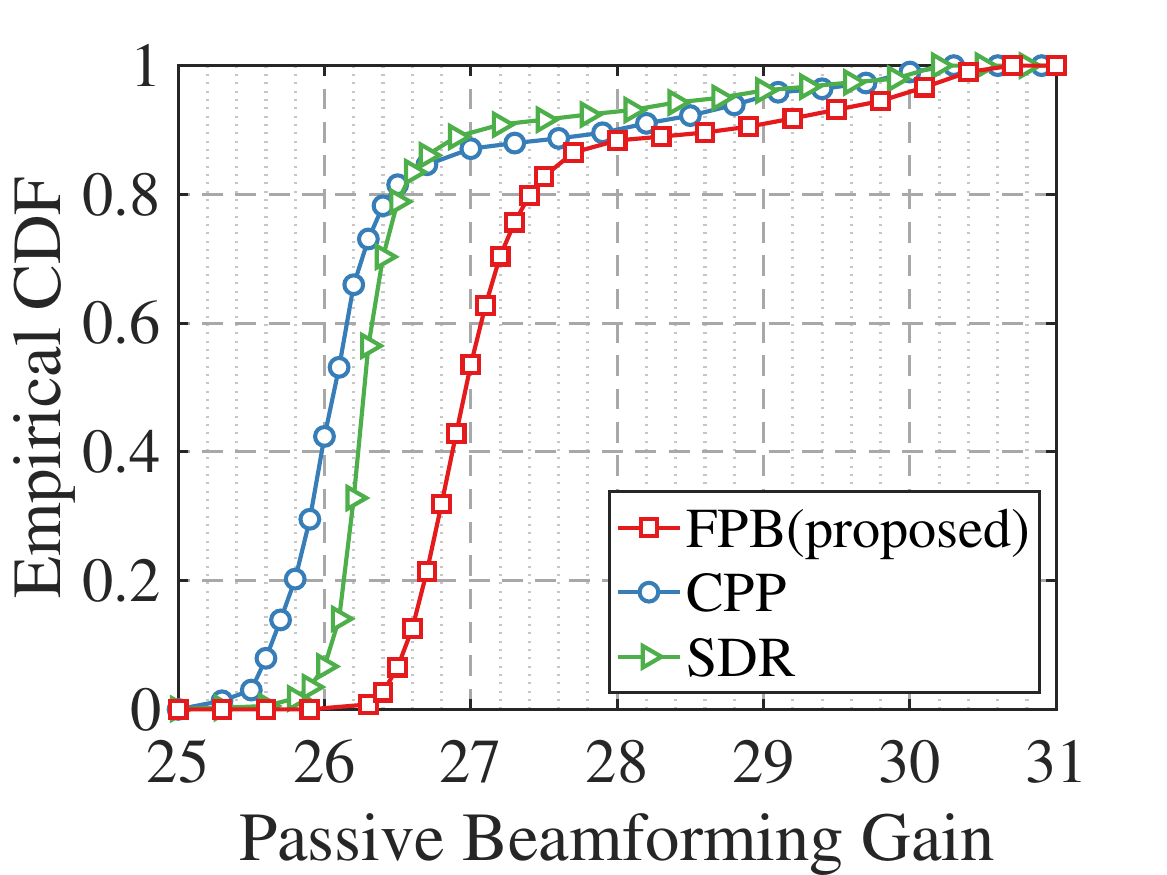}}
	\subfloat[$M=256,b=1$\label{fig: bf_cdf_b1_M256}]{
		\includegraphics[width=1.6in]{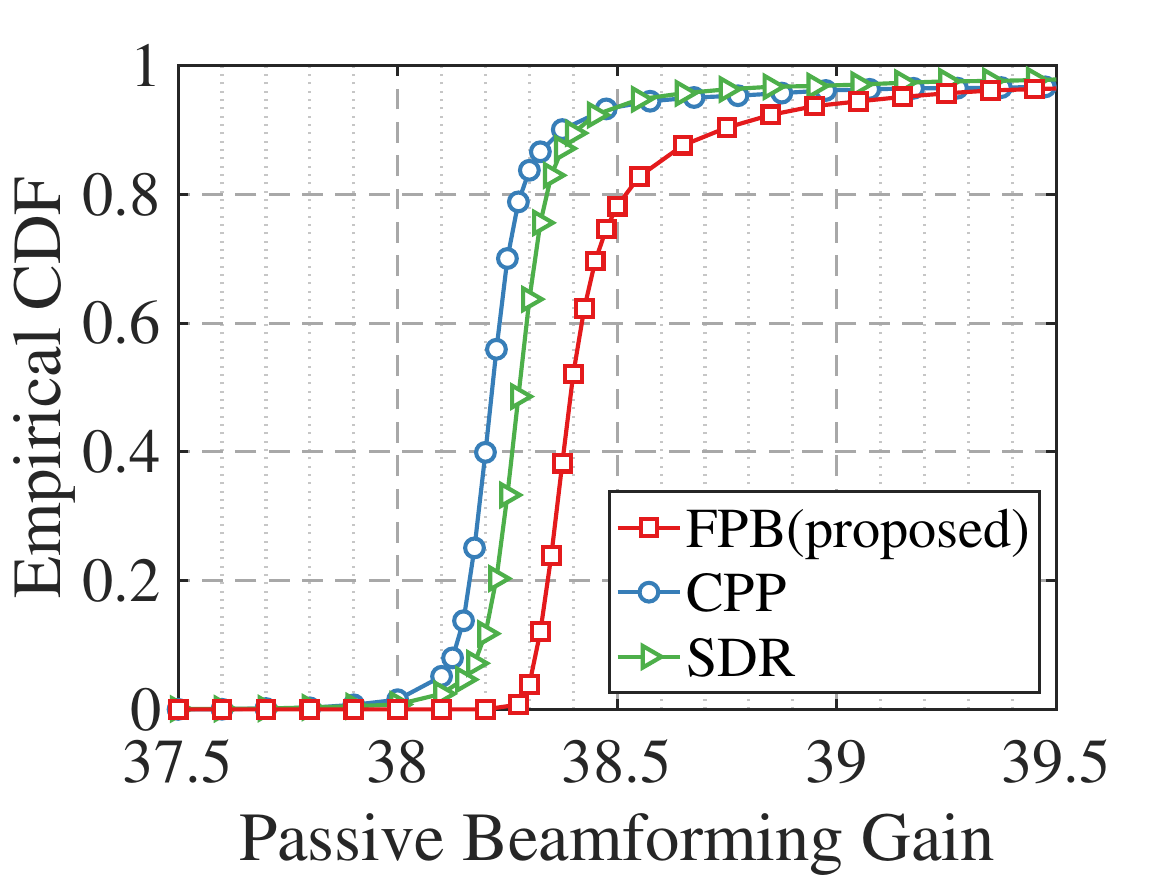}}
	
		\centering
	\subfloat[$M=64,b=2$\label{fig: bf_cdf_b2_M64} ]{
		\includegraphics[width=1.6in]{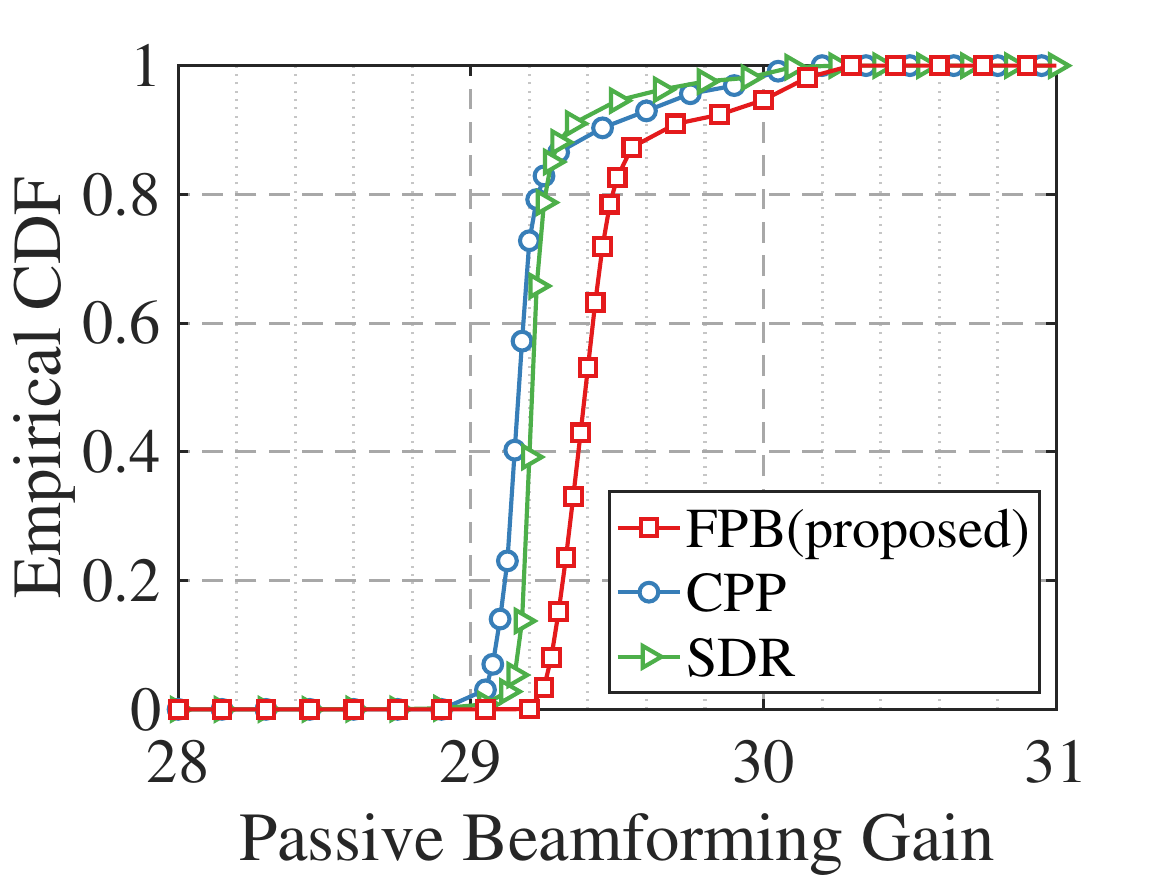}}
	\subfloat[$M=256,b=2$\label{fig: bf_cdf_b2_M256}]{
		\includegraphics[width=1.6in]{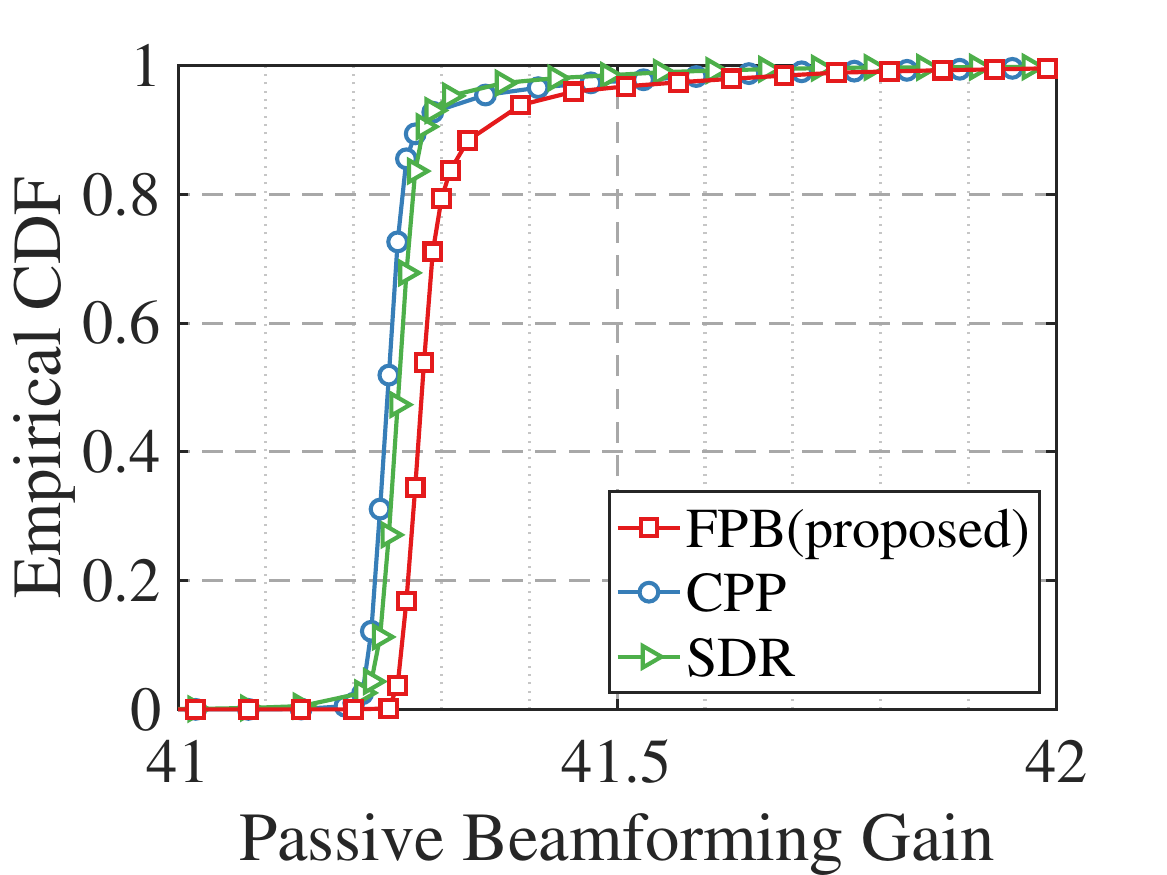}}
	\caption{Empirical CDF of the passive beamforming gain $10 \log_{10}F(\boldsymbol{\psi}^* )$ in dB.}
	\label{fig: bf_cdf_M}
\end{figure}

\begin{figure}[!t]
	\centering
	\subfloat[$b=1$\label{fig: bf_time_M1}]{
		\includegraphics[width=1.3in]{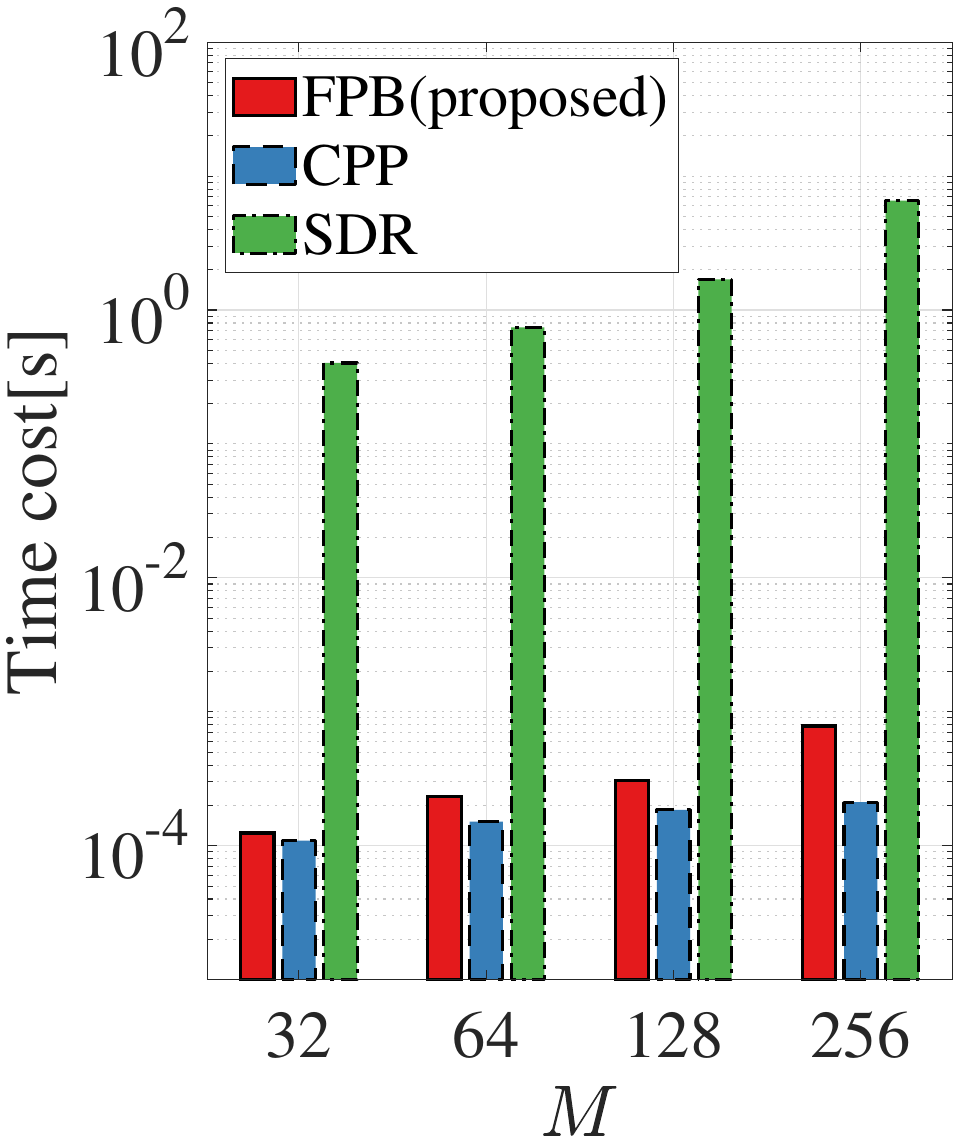}}
	\subfloat[$b=2$\label{fig: bf_time_M2}]{
		\includegraphics[width=1.3in]{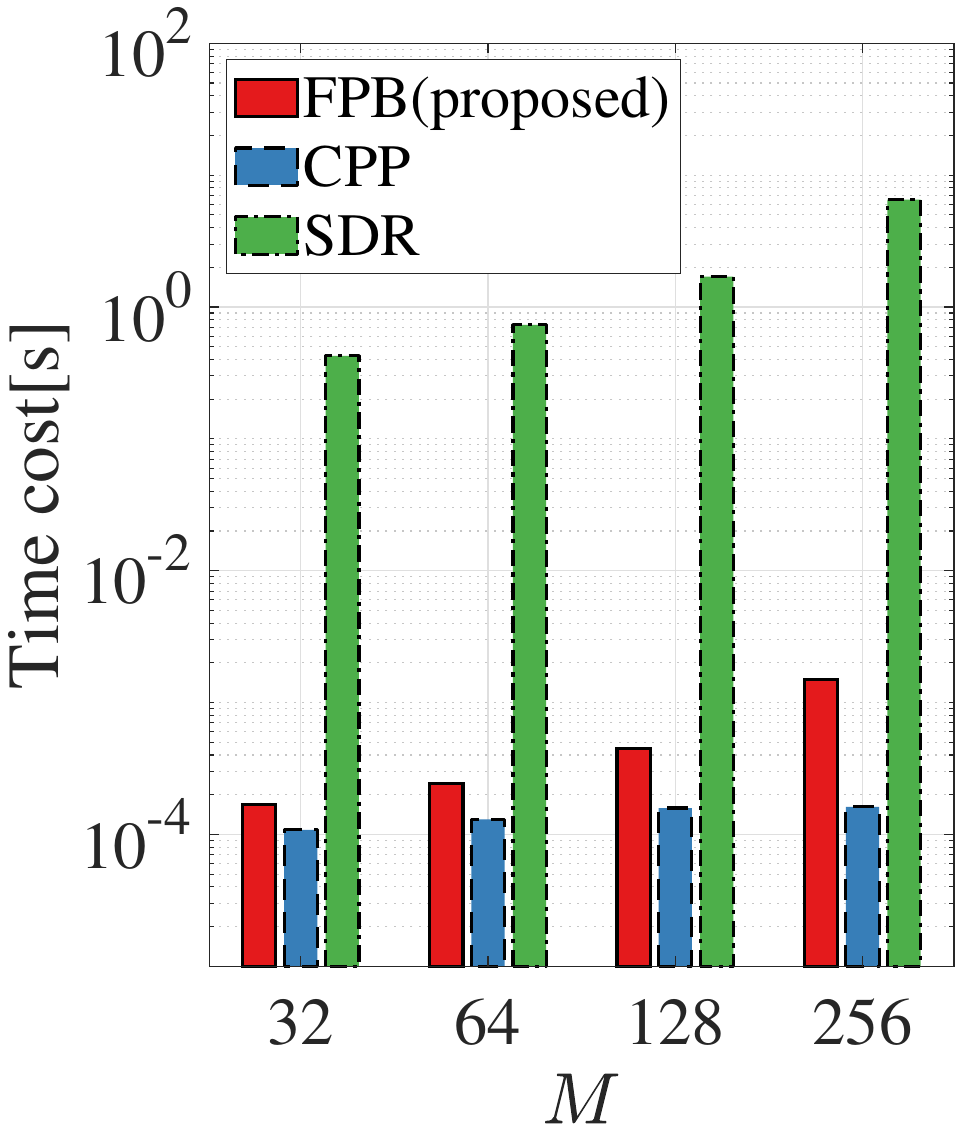}}
	\caption{Time cost of different passive beamforming algorithms.}
	\label{fig: bf_time_M}
\end{figure}

To evaluate the performance of the proposed FPB algorithm that optimally  solves the discrete beamforming problem \eqref{eqn: beamforming_problem}, we 
introduce the following two widely used methods as benchmarks:
\begin{itemize}	
	\item \emph{CPP method}\cite{wuBeamformingOptimizationWireless2020,youChannelEstimationPassive2020}:
	This method finds the RIS reflection coefficient for each segment $\ell=1,2,\ldots,L$  in the discrete set $\mathcal{F}^{M/L}$ that is closest to the end-to-end channel response $\bg_\ell^*$.
	\item \emph{SDR method}\cite{wuIntelligentReflectingSurface2019b,linReconfigurableIntelligentSurfaces2021}: 
	This method relaxes the original problem \eqref{eqn: beamforming_problem} to a continuous and rank-one-constrained semidefinite programming problem (SDP). The RIS reflection coefficient 
 for each segment $\ell=1,2,\ldots, L$  is the one in  $\mathcal{F}^{M/L}$ closest to the optimal solution of the SDP problem.
\end{itemize}

Suppose UEs are randomly distributed in the area of interest $\cM$. 
 For each UE position, we compute the passive beamforming gain $10 \log_{10}F(\boldsymbol{\psi}^\star )$,  as defined in \eqref{eqn: beamforming_problem}, and 
$\boldsymbol{\psi}^\star$ is the RIS reflection coefficient yielded by different RIS beamforming methods. 
Fig.~\ref{fig: bf_cdf_M} shows the empirical cumulative distribution function (CDF) of the passive beamforming gain $10 \log_{10}F(\boldsymbol{\psi}^\star )$ computed under different RIS size, i.e., $M\=64$ and $M\=256$, and different RIS phase bit, i.e., $b\=1$ and $b\=2$.  
The results of Fig.~\ref{fig: bf_cdf_M} demonstrate that the proposed FPB algorithm outperforms the CPP and SDR in terms of the passive beamforming gain, regardless of the RIS size and RIS phase bit. 
These results are consistent with the theoretical optimity of the FPB method, as guaranteed by Proposition~\ref{proposition: linear_complexity}.
Under the case $b\=1$, compared to the CPP method at the CDF $=\!80\%$, the FPB method achieves a gain of nearly $0.8$dB under the RIS scale $M\=64$ and about $0.2$dB gain under $M\=256$.
When the RIS bit number grows to $b\=2$, as shown in Fig.~\ref{fig: bf_cdf_b2_M64} and Fig.~\ref{fig: bf_cdf_b2_M256}, the gain brought by the proposed FPB method at CDF $=\!80\%$  decreases to be  $0.3$dB under the RIS scale $M\=64$ and less than $0.1$dB under $M\=256$.
This indicates that as the number of RIS elements increases, the performance of the CPP algorithm becomes closer to that of the proposed FPB method.
Since the FPB method yields the optimal solution, our work shows that the sub-optimal solution provided by the CPP algorithm can serve as a good approximation to the optimum when $b$ is large.

Fig.~\ref{fig: bf_time_M} compares the running time of the proposed FPB method and other benchmark methods under various RIS element numbers $M\=32, 64, 128, 256$.
Different bit number of RIS phase $b\=1,2$ is considered in Fig.~\ref{fig: bf_time_M1} and Fig.~\ref{fig: bf_time_M2}, respectively.
It is evident that the FPB and the CPP method exhibit superior time efficiency over the SDR method.
As shown in Fig.~\ref{fig: bf_time_M}, the FPB and the CPP methods take only a few milliseconds to operate the passive beamforming procedure, while the SDR method incurs an unacceptably high computational cost.
As the number of RIS elements $M$ increases from $M\=32$ to $M\=256$, the computational time of the CPP method linearly increases while the proposed FPB algorithm has a quadratic growth in time consumption.
The quadratic computational complexity of the FPB algorithm arises from the product of its linear search time in \eqref{prob: omega_beamforming2} and the linear complexity of computing the objective function value in \eqref{eqn: beamforming_problem}, which is consistent with the conclusion of Proposition~\ref{proposition: linear_complexity}. 
These results confirm the time efficiency of the proposed FPB algorithm, thus demonstrating the effectiveness of the proposed passive beamforming module.

Fig.~\ref{fig:peb_designed} shows the computed $\text{CRLB}(\bp)$ using (\ref{crlb:lower_bound_pos}) at $\text{SNR}\=6$dB and $M\=256$ for each position $\bp \!\in\! \cM$.
Specifically, Fig.~\ref{fig:peb_designed1} and Fig.~\ref{fig:peb_designed2} compare the $\text{CRLB}(\bp)$ obtained under the random RIS phases and the designed RIS phases by the FPB algorithm, respectively.
It is clear that the $\text{CRLB}(\bp)$ in Fig.~\ref{fig:peb_designed1}
is generally larger than the counterpart in Fig.~\ref{fig:peb_designed2}. By designing the RIS coefficients using the FPB algorithm, the $\text{CRLB}(\bp)$ is reduced from $1$ decimeter in Fig.~\ref{fig:peb_designed1} to less than $1$ centimeter in Fig.~\ref{fig:peb_designed2}. This means properly designing the RIS phase is important in harnessing the RIS benefit and improving the localization performance.
Moreover, as the position $\bp$ moves from the top left to the bottom right of the area $\cM$, the corresponding $\text{CRLB}(\bp)$ grows due to the increased distance between the UE and  RIS.

\begin{figure}[!t]
	\centering
	\subfloat[\label{fig:peb_designed1}]{
		\includegraphics[width=1.6in]{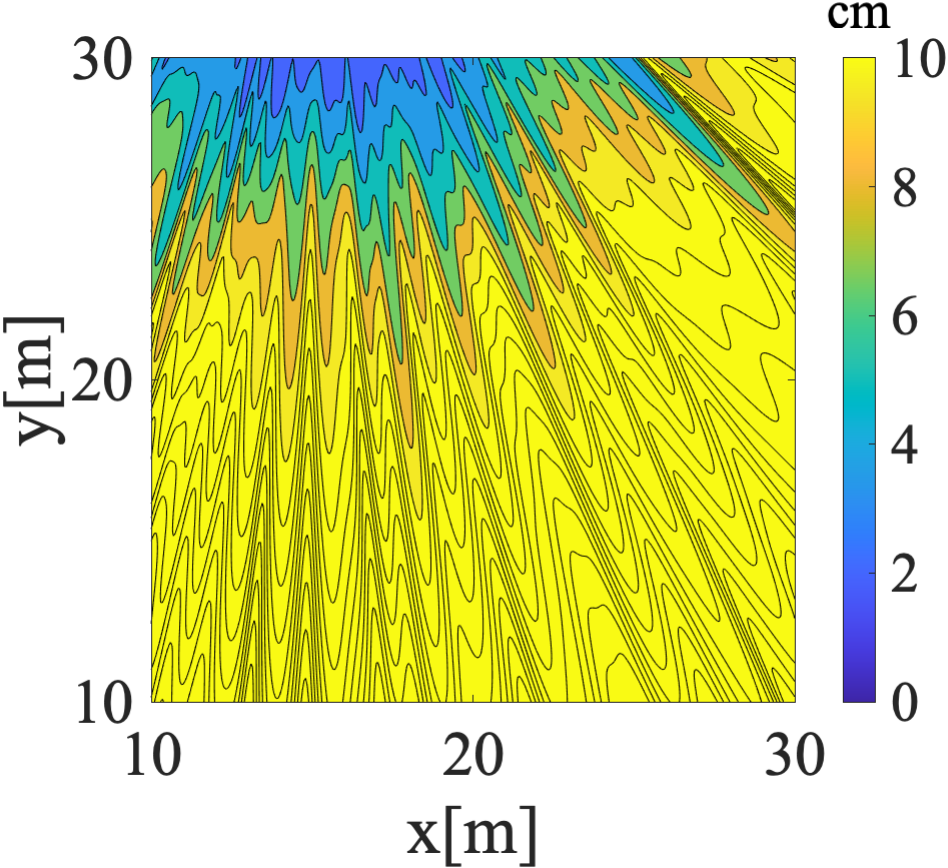}}
	\subfloat[\label{fig:peb_designed2}]{
		\includegraphics[width=1.6in]{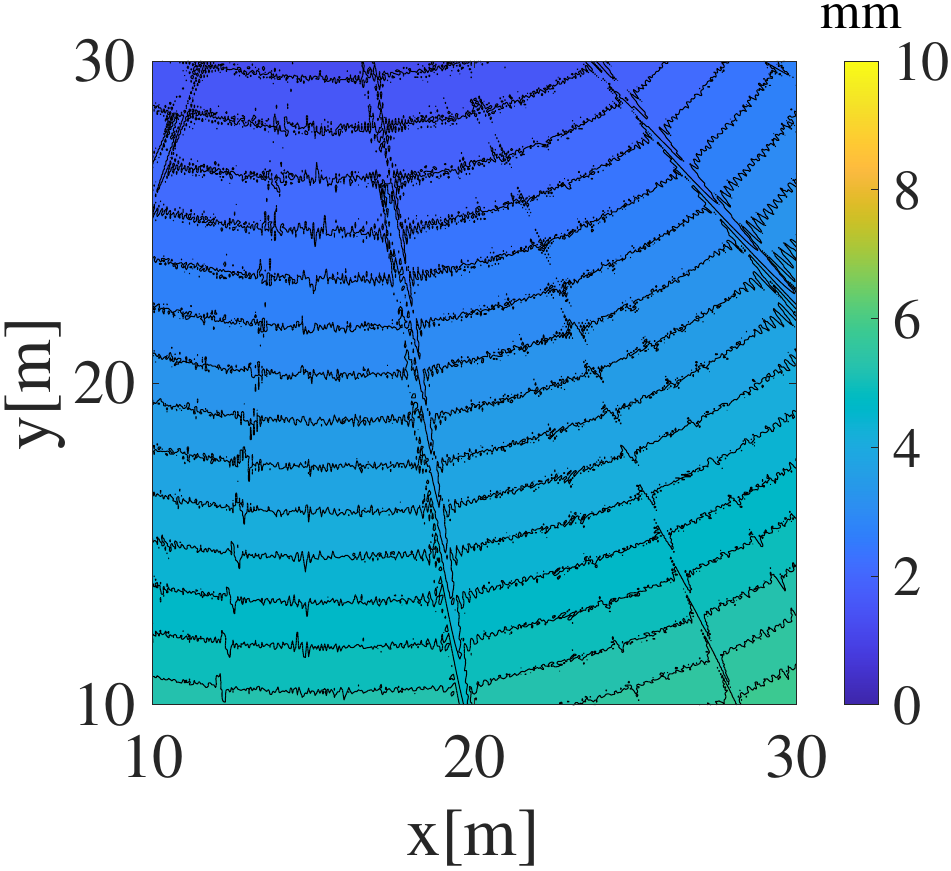}}
	\caption{$\text{CRLB}(\bp)$ under SNR $=\!6\text{dB}$ and: (a)Random RIS phases, (b)Designed RIS phases by FPB.}
	\label{fig:peb_designed}
\end{figure}

\subsection{Simulation Results: Coarse-to-Fine Localization}\label{subsec: result_loc}

\begin{figure}
	\centering
	\includegraphics[width=2in]{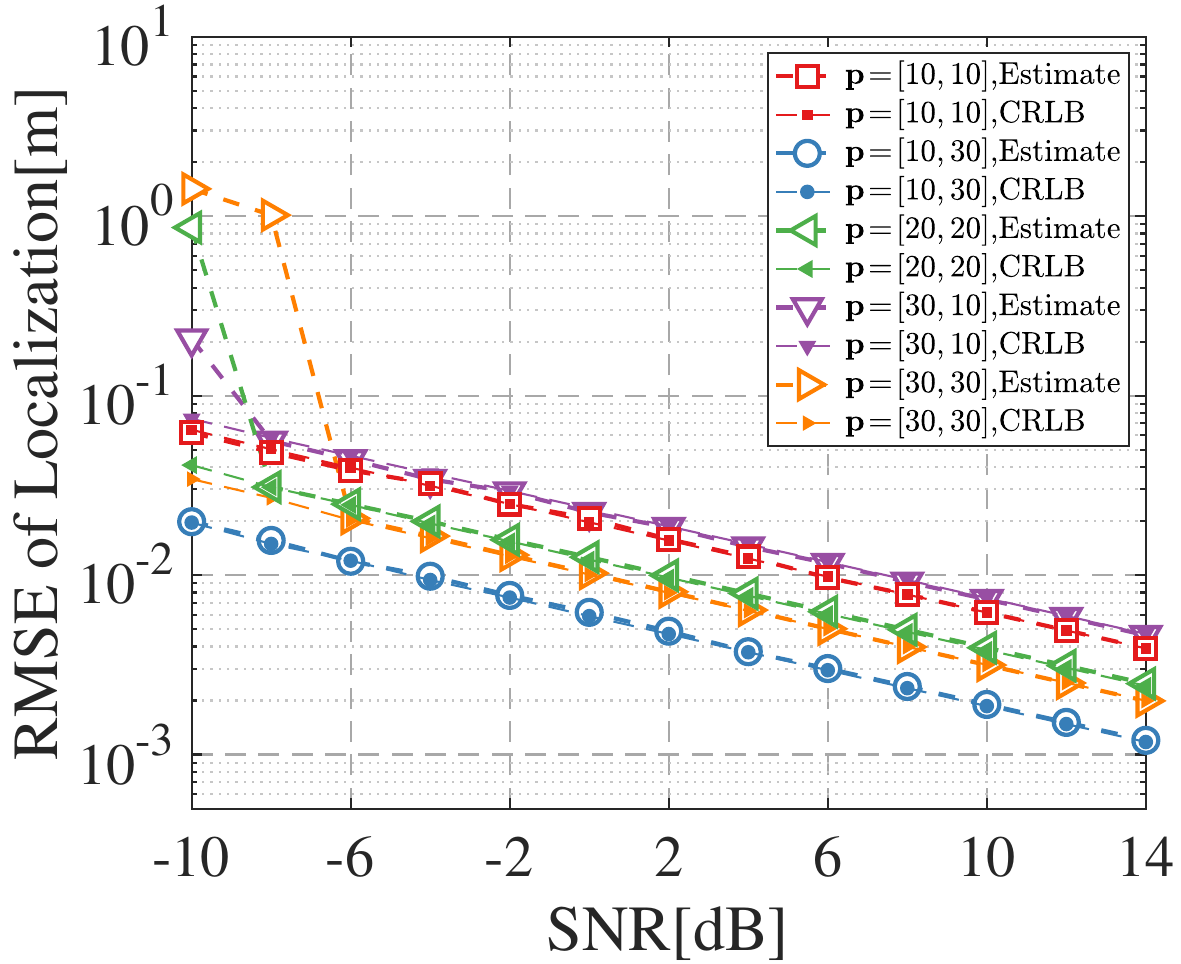}
	\caption{RMSE of the localization error.}
	\label{fig:fine_rmse_5points}
\end{figure}

\begin{figure}
	\centering
	\includegraphics[width=2in]{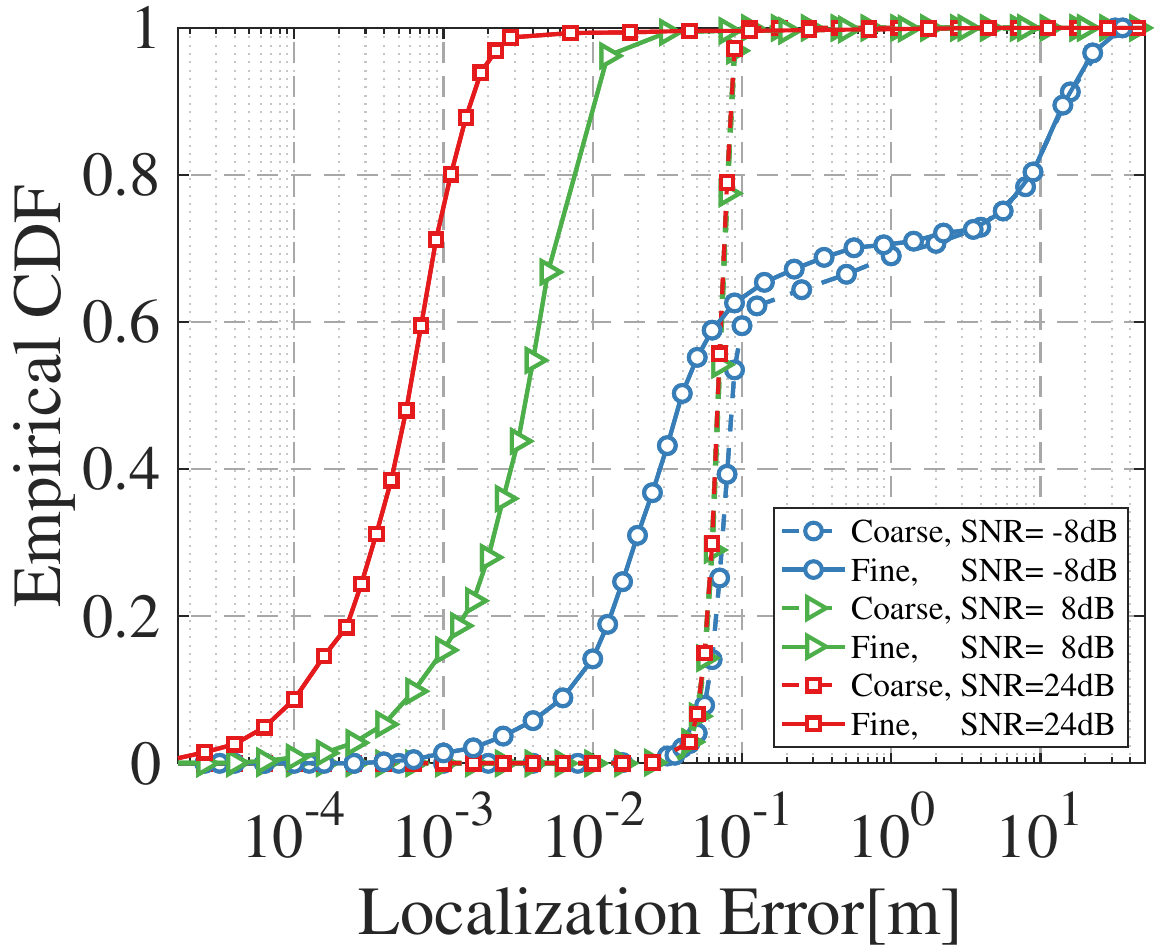}
	\caption{Empirical CDF of the coarse-to-fine localization error.}
	\label{fig:fine_cdf}
\end{figure}

\begin{figure}
	\centering
	\includegraphics[width=2in]{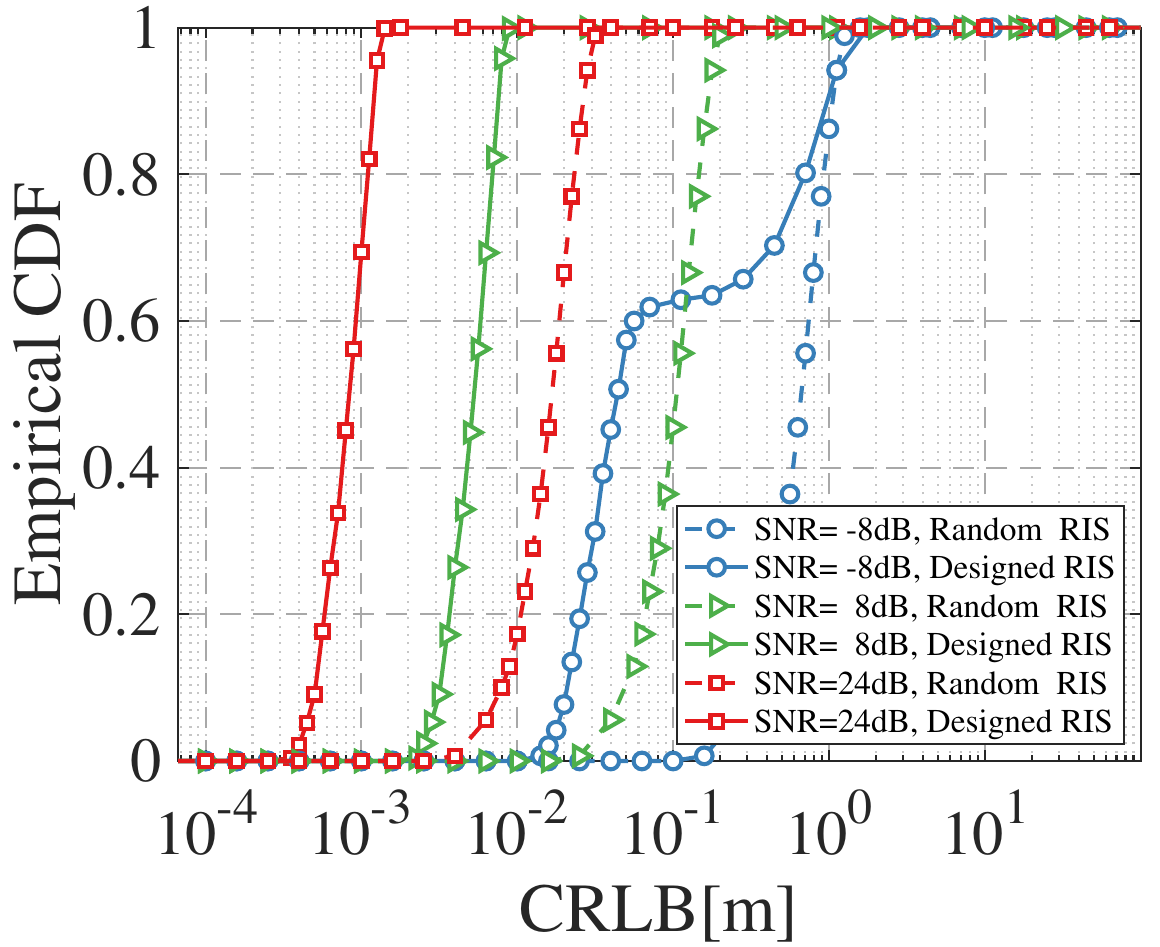}
	\caption{Empirical CDF of the $\text{CRLB}(\bp)$ over the entire area of interest.}
	\label{fig: result_crlb}
\end{figure}

In this subsection, we numerically validate the performance of the proposed coarse-to-fine localization algorithm.
At each  true position $\bp$, the localization error $\|\hat{\bp} - \bp\|_2$ and  RMSE $\sqrt{\mathbb{E}\|\hat{\bp} - \bp\|_2^2}$ 
are used as metrics to evaluate the localization performance, where $\hat{\bp}$ refers to the estimated UE position. 
The following simulation considers a general RIS setting with $M\=256$ and $b\=2$.

To demonstrate the gap between the RMSE $\sqrt{\mathbb{E}\|\hat{\bp} - \bp\|_2^2}$ and its theoretical bound $\text{CRLB}(\bp)$ in \eqref{crlb:lower_bound_pos}, Fig.~\ref{fig:fine_rmse_5points} compares the 
RMSE $\sqrt{\mathbb{E}\|\hat{\bp} - \bp\|_2^2}$ and the corresponding $\text{CRLB}(\bp)$ 
at five specific points under different SNRs. 
The five points include  four corner points and one center point of the area of interest $\cM$, i.e., $\bp\= (10,10)$m, $(10,30)$m, $(30,30)$m, $(30,10)$m, and $(20,20)$m.
It is observed that the proposed coarse-to-fine localization result approaches the CRLB when $\text{SNR}\!\geq\! -6$dB. 
The high positioning accuracy in a relatively low SNR, i.e., $\text{SNR} \=-6$dB is achieved due to the high RIS beamforming gain illustrated in Fig.~\ref{fig: bf_cdf_M}.

To evaluate the coarse and fine localization performance over the entire area of interest $\cM$, Fig.~\ref{fig:fine_cdf} shows the empirical CDF of the coarse and the fine localization errors.
Specifically, under varying SNR=$-8, 8, 24$dB, we consider $1000$ independent UEs uniformly selected from the $\cM$ and estimate these UEs' positions by the proposed coarse-to-fine localization scheme.
In the case of low SNR ($-8$dB), we observe that around $40\%$ of the UEs have a  coarse localization error exceeding $1$ decimeter, and there is almost no improvement in the fine localization procedure. 
In contrast, the fine localization module can effectively refine the coarse estimates for the remaining $60\%$ UEs whose coarse localization errors are less than $1$ decimeter. 
This is because an accurate coarse localization result provides a good initialization for the gradient-based fine localization module.
As shown in Fig.~\ref{fig:fine_cdf}, the error of coarse localization remains unchanged 
when the SNR increases higher than $8$dB.
This phenomenon arises due to the consistent resolution for extracting the ToA $\hat{\tau}$ and AoA $\hat{f}$ from problems \eqref{pro:toa_fft} and \eqref{pro:aoa_fft} when the SNR is sufficiently high.
In the case of median SNR ($8$dB), more than $90\%$ of the UEs exhibit a localization error of less than $1$cm.
Besides, over $80\%$ of the UEs exhibit a localization error of less than $1$ millimeter under high SNR ($24$dB), which validates the high accuracy of the proposed coarse-to-fine localization protocol.

To further demonstrate the estimation performance of the proposed localization algorithm, Fig.~\ref{fig: result_crlb} shows the CDF of the corresponding $\text{CRLB}(\bp)$ for each sampled UE $\bp$ over the area of interest $\cM$.
At  SNR=$-8$dB, the localization accuracy of the coarse estimation model is low in Fig.~\ref{fig:fine_cdf}. As a result, 
the passive beamforming module takes the inaccurate coarse localization results as the input and fails to optimize the RIS phases, leading to no significant improvement in the $\text{CRLB}(\bp)$.
However, as the SNR increases to $8$dB and the coarse localization accuracy improves, the $\text{CRLB}(\bp)$ with designed RIS phases is substantially lower than that with random RIS phases.
At median SNR ($8$dB), the $\text{CRLB}(\bp)$ is approximately at the centimeter level, while at high SNR ($24$dB), the $\text{CRLB}(\bp)$ can be reduced from centimeter level to millimeter level.
The theoretical results in Fig.~\ref{fig: result_crlb} are consistent with the numerical simulation in Fig.~\ref{fig:fine_cdf}, suggesting that the proposed two-stage localization algorithm has achieved near-optimal localization accuracy.


\section{Conclusions} \label{Conclusions}
In this paper, a large-scale RIS is used to aid UE localization in a near-field mmWave MIMO-OFDM system.
We proposed a partitioned-far-field architecture to address the spatial non-stationarity of the near-field localization and the scalability issue of discrete beamforming.
Based on this architecture, we developed a near-field localization protocol consisting of the balanced signaling module, the passive beamforming module, and the two-stage localization module.
The balanced signaling method separates LoS and NLoS components in the raw signal and formulates the localization problem.
The proposed linear-time FPB algorithm optimally solves the discrete beamforming problem and improves the theoretical localization performance.
By exploiting the channel sparsity of the proposed partitioned-far-field representation, we also devised a two-stage coarse-to-fine localization algorithm to estimate the UE's position. 
Simulation results demonstrated the efficiency of the proposed discrete beamforming and coarse-to-fine localization algorithms, with localization error reduced from decimeter-level to centimeter-level. 
Overall, our approach has provided a low-complexity and high-precision protocol for RIS-assisted near-field localization.



\begin{appendices}

\section{Proof of Proposition \ref{proposition: crlb_and_powergian}} \label{appendix-C}
Let $\kappa(\bI(\bp ))$ be the condition number of FIM $\bI(\bp )$ in \eqref{eqn:FIM_2} and define $M_0 := \max_{\bp } \kappa(\bI(\bp ))$.
Then, using the closed form of $ \tilde{\dY}_{n,t,:}$ and $\bz_{n,\ell}$ from \eqref{eqn: z_n}, we have
{\footnotesize 
\begin{multline} 
  \tr\left\lbrace \bI^{-1}\left(\bp \right) \right\rbrace 
 \overset{(a)}{\le}   2 \sigma^2 M_0\left[\sum_{n=1}^{N}\sum_{t=1}^{T} \sum_{d=1}^2 \left\| \frac{\partial \tilde{\dY}_{n,t,:} }{\partial p_d}\right\|_2^2\right]^{-1}     \\
 =  2 \sigma^2 M_0\left[\sum_{n=1}^{N}\sum_{t=1}^{T} \sum_{d=1}^2 \left\| \sum_{\ell = 1}^L   \left( \boldsymbol{\psi}_{t,\ell}^\rt \bg_{\ell}  \right)    \frac{\partial \bz_{n,\ell} }{\partial  p_d}  
\right.\right.    \\
  \left. \left. -  j   	\left[ \boldsymbol{\psi}_{t,\ell}^\rt \left(   \bb_{\frac{M}{L}}   \odot  \bg_{\ell} \right)    \right] 
\left( \pi \frac{\partial  \cos \theta^\text{R,M}_\ell   }{\partial p_d}   \bz_{n,\ell} \right)
\right\|_2^2\right]^{-1}    \\
 \overset{(b)}{\le}
    2 \sigma^2 M_0\left[\sum_{n=1}^{N}\sum_{t=1}^{T} \sum_{d=1}^2 \left\| \sum_{\ell = 1}^L   \left( \boldsymbol{\psi}_{t,\ell}^\rt \bg_{\ell}  \right)  \frac{\partial \bz_{n,\ell} }{\partial  p_d}
\right\|_2^2\right]^{-1},
\label{eqn: prop3_ineq1}
\end{multline}
}%
where $\bb_{N} \= [0,1,\ldots,N-1 ]^\rt $.
The inequality (a) holds because $\tr\lbrace \bI(\bp) \rbrace \tr\lbrace \bI^{-1} (\bp )\rbrace \= \sum_{d=1}^2 \lambda_d \sum_{d=1}^2\lambda_d^{-1}\!\le\! 2+2 \kappa(\bI(\bp )) \le 4 M_0$ and $\{\lambda_d\}_d$ are the eigenvalues of $\bI (\bp )$.
The last inequality in (b) relies on the fact that $  \|  z_1 + z_2  \| _2 \ge  \left\|z_1\right\|_2$ holds for $A(z_1,z_2)\le \pi/2$.
Moreover, we define $\epsilon_0 \= \min_{\bp} \| \boldsymbol{\psi}_{t,\ell}^\rt \bg_{\ell}   \|_2$ where $\epsilon_0 \!>\! 0$ holds due to the constraint \eqref{eqn: crlb_power_mild}. 
Using the reversed version of Cauchy-Schwarz inequality and the observation that $\epsilon_0^2 \le \| \boldsymbol{\psi}_{t,\ell}^\rt \bg_{\ell} \|_2^2 \le M/L$, we have
\begin{equation}\footnotesize
\begin{aligned}
 & \left\|  \sum_{\ell = 1}^L   \left( \boldsymbol{\psi}_{t,\ell}^\rt \bg_{\ell}  \right)  \frac{\partial \bz_{n,\ell} }{\partial  p_d} \right\|_2^2  \ge 4\left( m_{n,t,d}^{\frac{1}{2}}+ m_{n,t,d}^{-\frac{1}{2}}\right) ^{-2} \\   
&  \times \left(   \sum_{\ell = 1}^L   \left\| \frac{\partial \bz_{n,\ell} }{\partial  p_d}   \right\|_2^2 \right) \left( \sum_{\ell = 1}^L   \left\| \boldsymbol{\psi}_{t,\ell}^\rt \bg_{\ell}    \right\|_2^2\right) \treq C_1\sum_{\ell = 1}^L   \left\| \boldsymbol{\psi}_{t,\ell}^\rt \bg_{\ell}    \right\|_2^2 ,
\end{aligned}\label{eqn: prop3_ineq2} 
\end{equation}
where $m_{n,t,d} \= \epsilon_0\sqrt{L/M} (\min_{\ell} \| \frac{\partial \bz_{n,\ell} }{\partial  p_d}  \|_2) (\max_{\ell} \| \frac{\partial \bz_{n,\ell} }{\partial  p_d} \|_2)^{-1}$ and $m_{n,t,d}\!>\!0$ is also guaranteed by the constraint \eqref{eqn: crlb_power_mild}.
In \eqref{eqn: prop3_ineq2}, $C_1\!>\!0$ is a constant independent of the RIS phase control $\{\boldsymbol{\psi}_{t,\ell}\}_{t,\ell}$ and UE's position $\bp$.
By combining the results of \eqref{eqn: prop3_ineq1} and \eqref{eqn: prop3_ineq2}, the proof is completed.

\qed

\section{Proof of Lemma \ref{lemma: opt_condition}} \label{appendix-D}
We prove this lemma by contradiction. 
Let $\{\boldsymbol \psi_{\ell}^\star\}_{\ell=1}^L$ be the optimal solution of problem \eqref{eqn: beamforming_problem}.
We assume that there exists a partition $\{S_1,S_2\}$ of $\{1,2,\ldots,M/L\}$ and a $\ell_0 \in \{1,2,\ldots,L\}$ such that
\small\begin{equation*}
 A\bigg( \sum_{k \in S_1}\psi^\star_{\ell_0,k} g_{\ell,k} ,\sum_{k \in S_2} \psi^\star_{\ell_0,k} g_{\ell,k} \bigg) > \frac{\pi}{2^b}.
\end{equation*}\normalsize

Given any $I_1,I_2 \in \cC$, note that the collection of all points in $\{\psi I_2: \psi\in\mathcal{F} \}$ (namely, all the candidate rotation of $I_2$) is uniformly located on a circle with radius $\left| I_2 \right|$, where the angle between two adjacent points is exactly equal to $\pi/2^b$.
Therefore, there must exist a possible rotation $\psi_0 \in \mathcal{F}$ satisfying that the angle between $\psi_0 I_2$ and $I_1$ is no more than $\pi/2^b$.
Specifically, let $I_1\= \sum_{k \in S_1}\psi^\star_{\ell_0,k} g_{\ell,k}$ and $I_2\= \sum_{k \in S_2} \psi^\star_{\ell_0,k} g_{\ell,k}$, and we have
\small\begin{equation*}
 A\bigg( \sum_{k \in S_1}\psi^\star_{\ell_0,k} g_{\ell,k} ,\sum_{k \in S_2} \psi_0\psi^\star_{\ell_0,k} g_{\ell,k} \bigg) \le \frac{\pi}{2^b}.
\end{equation*}\normalsize
where $\psi_0\psi^\star_{\ell_0,k} \in \mathcal{F}$ because $\mathcal{F}$ is a closed set under multiplication.

In what follows, we use this $\psi_0$ to construct a new solution $\{\boldsymbol \psi_{\ell} \}_{\ell=1}^L$ that has a smaller object value than $\{\boldsymbol \psi_{\ell}^\star\}_{\ell=1}^L$, which leads to a contradiction.
We consider a different RIS phase design $\{\boldsymbol \psi_{\ell} \}_{\ell=1}^L \neq \{\boldsymbol \psi^\star_{\ell} \}_{\ell=1}^L$ such that$\boldsymbol \psi_{\ell_0, k} \=\psi_0 \boldsymbol \psi^\star_{\ell_0, k} $ for $k \in S_1$, $\boldsymbol \psi_{\ell_0, k} \= \boldsymbol \psi^\star_{\ell_0, k} $ for $k \in S_2$ and $\boldsymbol \psi_{\ell} \=\boldsymbol \psi^\star_{\ell} $ for $\ell \neq \ell_0$.
A straightforward comparison reveals that
\begin{equation*}\small
\begin{aligned}
 &\sum_{\ell=1}^L \left|   \boldsymbol  \psi_\ell^\rt \bg_\ell \right|^2 - \sum_{\ell=1}^L \left|   {(\boldsymbol  \psi^\star_\ell)}^\rt \bg_\ell \right|^2 = \small \left|   I_1 + \psi_0 I_2 \right|^2 - \left |  I_1 + I_2 \right|^2 \\ 
&= 2 \left| I_1 \right|  \left|I_2  \right| \left[   \cos\left(A\left( I_1 ,\psi_0 I_2 \right) \right) - \cos\left(A\left( I_1 , I_2 \right) \right) \right]   > 0.
\end{aligned}
\end{equation*}
The last inequality holds because $\cos\left(A\left( I_1 ,\psi_0 I_2 \right) \right) \!\ge\! \cos \frac{\pi}{2^b} \!>\! \cos\left(A\left( I_1 ,I_2 \right) \right)$. 
Then, $\small \sum_{\ell=1}^L |   \boldsymbol  \psi_\ell^\rt \bg_\ell |^2 \!>\! \sum_{\ell=1}^L |   {(\boldsymbol  \psi^\star_\ell)}^\rt \bg_\ell |^2$ indicates that $\{\bpsi_{\ell}^\star\}_{\ell=1}^L$ is not optimal, which is a contradiction and completes the proof.
\qed

\section{Proof of Lemma \ref{lemma: equivalence_omega_beamformer}} \label{appendix-E}
For ease of exposition, we define the following function
\small\begin{equation*}
~~~\hat{\psi}_{\ell,k}(\omega)
= \exp \left\lbrace j\frac{2\pi}{2^b} \mathrm{round} \left(-\frac{2^b}{2\pi}\mathrm{arg}(g_{\ell,k}) + \frac{2^b}{2\pi}\omega  \right) \right\rbrace,
\end{equation*}\normalsize
where $\omega\in[0,2\pi)$ and $g_{\ell,k}$ is the $k$-th element of the the combined channel response $\bg_\ell$.
Our objective is to prove that ${\psi}_{\ell,k} = \hat{\psi}_{\ell,k}(\omega_\ell)$ holds for all $\ell,k$.
This can be realized by estimating the upper bound of $A ( {\psi}_{\ell,k} ,  \hat{\psi}_{\ell,k}(\omega_\ell) )$ and showing that $A ( {\psi}_{\ell,k} ,  \hat{\psi}_{\ell,k}(\omega_\ell) )\= 0$.

As $\boldsymbol \psi_{\ell}$ satisfies the constraint \eqref{eqn: opt_condition} in Lemma \ref{lemma: opt_condition}, we have
\small\begin{equation*}
A ( {\psi}_{\ell,k} g_{\ell,k}, \sum_{m} {\psi}_{\ell,m} g_{\ell,m}  ) \!<\! A ( {\psi}_{\ell,k} g_{\ell,k}, \sum_{m\neq k } {\psi}_{\ell,m} g_{\ell,m}  ) \!\le\! \frac{\pi}{2^b}.
\end{equation*}\normalsize
On the other hand,  
$A(\hat{\psi}_{\ell,k}(\omega) g_{\ell,k},\exp(j\omega) \le {\pi} / {2^b}$ holds for all $\omega \in [0,2\pi)$ because the inequality $|  \mathrm{round}(x) - x  | \le 1/2$ holds for $x>0$.
Note that $\omega_\ell = \mathrm{arg}(\sum_{m} {\psi}_{\ell,m} g_{\ell,m})$, so we can finally estimate the value of $A ( {\psi}_{\ell,k} ,  \hat{\psi}_{\ell,k}(\omega_\ell))$ by
\small\begin{align*}
&   A  ( {\psi}_{\ell,k} ,  \hat{\psi}_{\ell,k}(\omega_\ell)  )  = A  ( {\psi}_{\ell,k} g_{\ell,k} ,  \hat{\psi}_{\ell,k}(\omega_\ell) g_{\ell,k} ) \\ 
\le&    A  ( {\psi}_{\ell,k} g_{\ell,k} ,  \sum_{m} {\psi}_{\ell,m} g_{\ell,m}   ) + A (  \sum_{m} {\psi}_{\ell,m} g_{\ell,m} ,\hat{\psi}_{\ell,k}(\omega_\ell) g_{\ell,k} )\\ 
\le&   A  ( {\psi}_{\ell,k} g_{\ell,k} ,  \sum_{m} {\psi}_{\ell,m} g_{\ell,m}   ) + A ( \exp(j\omega_\ell) ,\hat{\psi}_{\ell,k}(\omega_\ell) g_{\ell,k} )\\ 
< &   \frac{\pi}{2^b} + \frac{\pi}{2^b} \le \frac{2\pi}{2^b}.
\end{align*}\normalsize
This implies that $A ( {\psi}_{\ell,k} ,  \hat{\psi}_{\ell,k}(\omega_\ell) ) \= 0$ as $ {\psi}_{\ell,k}$ and $\hat{\psi}_{\ell,k}(\omega_\ell)$ are multiples of $\exp\left(j {2\pi} / {2^b} \right)$, which completes the proof.
\qed

\section{Proof of Proposition \ref{proposition: linear_complexity}} \label{appendix-F}
According to the result of Lemma~\ref{lemma: opt_condition} and Lemma~\ref{lemma: equivalence_omega_beamformer}, the  optimal solution $\{\boldsymbol{\psi}_\ell^\star\}_{\ell=1}^L$ of problem \eqref{eqn: beamforming_problem} is contained in the $\omega$-beamformer set $\{{\boldsymbol \psi}_{\ell}(\omega),\omega \!\in\! [0,2\pi)\}$. 
In what follows, we show how to solve the problem \eqref{prob: omega_beamforming} by searching for $\{ \omega^\star_\ell \}_{\ell=1}^L$ with at most $\mathcal{O}(M)$ complexity.

Note that $ { \boldsymbol  \psi}_{\ell}^\rt(\omega) \bg_\ell = \sum_{k=1}^{M/L} { \psi}_{\ell,k}(\omega)g_{\ell,k} $ is a step function with respect to $\omega$, i.e., piece-wise constant. 
Therefore, in order to find the maximum of function $\| { \boldsymbol  \psi}_{\ell}^\rt(\omega) \bg_\ell\|^2 \=|\sum_{k=1}^{M/L} { \psi}_{\ell,k}(\omega)g_{\ell,k}|^2 $, we only need to evaluate those $\omega$s that are discontinuous points of ${ \psi}_{\ell,k}(\omega)$. 
Since the set of all discontinuous points of function $\mathrm{round}(\cdot)$ is $\{m + \frac{1}{2},m\in\cZ\}$, we search for the discontinuous points of ${ \psi}_{\ell,k}(\omega)$ by solving the following inequality
\small\begin{eqnarray*}
 ~~-\frac{2^b}{2\pi}\mathrm{arg}(g_{\ell,k}) + \frac{2^b}{2\pi}\omega = m_{\ell,k}+\frac{1}{2} , 0 \le \omega < 2 \pi,m_{\ell,k}\in \cZ.
\end{eqnarray*}\normalsize
Here, the solution is $\omega = \mathrm{arg}(g_{\ell,k}) + 2^{-b}(2\pi m_{\ell,k} + \pi)$, and the integer $m_{\ell,k}$ satisfies
\small\begin{eqnarray*}
~-\frac{2^b \cdot \mathrm{arg}(g_{\ell,k})}{2 \pi} - \frac{1}{2} \le m_{\ell,k} < -\frac{2^b \cdot \mathrm{arg}(g_{\ell,k})}{2 \pi} - \frac{1}{2} + 2^b,
\end{eqnarray*}\normalsize
which completes the proof.
\qed

\end{appendices}

\bibliographystyle{IEEEtran}
\bibliography{myRef}

\end{document}